\pdfoutput=1

\documentclass[sigconf,nonacm]{acmart}

\makeatletter                   
\def\mdseries@tt{m}             
\makeatother                    
\usepackage[plain]{fancyref}
\usepackage[draft=true]{minted} 
\usepackage{color}
\usepackage{hyperref}           
\hypersetup{
    colorlinks=true,
    linkcolor=blue,
    filecolor=red,      
    urlcolor=magenta,
    breaklinks=true,            
}
\usepackage{breakurl}           

\usepackage{balance} 
\usepackage{enumitem}
\usepackage{subfigure}
\usepackage{tikz}
\usepackage{pgfplots}
\pgfplotsset{compat=newest}

\usetikzlibrary{patterns}
\usepgfplotslibrary{statistics}

\title[Tractable mechanisms for computing near-optimal utility functions]{Tractable mechanisms for computing\\near-optimal utility functions}

\author{Rahul Chandan}
\affiliation{
  \institution{UC Santa Barbara}
  \city{Santa Barbara}
  \state{California}}
\email{rchandan@ucsb.edu}

\author{Dario Paccagnan}
\affiliation{
  \institution{Imperial College, London}
  \city{London, United Kingdom}}
\email{d.paccagnan@imperial.ac.uk}

\author{Jason R. Marden}
\affiliation{
  \institution{UC Santa Barbara}
  \city{Santa Barbara}
  \state{California}}
\email{jrmarden@ece.ucsb.edu}

\begin{abstract}
Large scale multiagent systems must rely on distributed decision making, as centralized coordination is either impractical or impossible. Recent works approach this problem under a game theoretic lens, whereby utility functions are assigned to each of the agents with the hope that their local optimization approximates the centralized optimal solution. Yet, formal guarantees on the resulting performance cannot be obtained for broad classes of problems without compromising on their accuracy. In this work, we address this concern relative to the well-studied problem of resource allocation with nondecreasing concave welfare functions. We show that optimally designed local utilities achieve an approximation ratio (price of anarchy) of $1-c/e$, where $c$ is the function's curvature and $e$ is Euler's constant. The upshot of our contributions is the design of approximation algorithms that are distributed and efficient, and whose performance matches that of the best existing polynomial-time (and centralized) schemes.  
\end{abstract}


\newcommand{\BibTeX}{\rm B\kern-.05em{\sc i\kern-.025em b}\kern-.08em\TeX}
\theoremstyle{plain}
\newtheorem{theorem}{Theorem}
\newtheorem{lemma}{Lemma}
\newtheorem{proposition}{Proposition}
\newtheorem{corollary}{Corollary}

\theoremstyle{definition}
\newtheorem{definition}[]{Definition}

\theoremstyle{remark}
\newtheorem{example}[]{Example}

\newcommand{\mc}[1]{\mathcal{#1}}
\newcommand{\mb}[1]{\mathbb{#1}}
\newcommand{\opt}[1]{{#1}^\mathrm{opt}}
\newcommand{\nash}[1]{{#1}^\mathrm{ne}}
\newcommand{\reals}{\mb{R}}
\newcommand{\naturals}{\mb{N}}

\DeclareMathOperator*{\argmax}{arg\,max}

\newcommand{\poa}{{\rm PoA}}

\begin{document}


\pagestyle{fancy}
\fancyhead{}


\maketitle

\section{Introduction}

The study of distributed control in multiagent systems has gained popularity over the
past few decades as it has become apparent that the behaviour of local decision makers 
impacts the performance of many social and technological systems.
Consider the typical example of selfish drivers on a road network.
Counterintuitively, if all drivers make route selections that minimize their own travel times,
the average time each driver spends on the road can be much higher than optimal~\cite{ccolak2016understanding}.
As an alternative example, DARPA's Blackjack program aims to launch
satellite constellations with a high degree of mission-level autonomy
into low Earth orbit~\cite{forbes2020blackjack}.
The key objectives of the Blackjack program include developing the on-orbit,
distributed decision making capabilities within these satellite networks,
as agent coordination cannot rely upon the unreliable and high latency 
communications from ground control.

In either of the scenarios described above, the system would perform most efficiently
if a central coordinator could compute and relay the optimal decisions to each of the agents.
However, in the systems we have discussed, coordination by means of a central authority
is either impractical -- due to latencies and bandwidth limitations in communications,
scalability and security requirements, etc. --
or even impossible (e.g., dictating what route each driver must follow is not presently possible).
Thus, in these systems, decision making \emph{must} be distributed.
The inevitable loss in performance when coordination is distributed --
often referred to as the ``tragedy of the commons'' in economics and environmental sciences~\cite{lloyd1833two,hardin1968tragedy} --
is well-documented in many scenarios~\cite{adilov2015economic,ccolak2016understanding,giubilini2019antibiotic}.
Evidently, the design of algorithms that mitigate the losses in system performance 
stemming from the distribution of decision making is critical to 
the implementation of the multiagent systems described.

A fruitful paradigm for the design of distributed multiagent coordination algorithms
-- termed the \emph{game theoretic approach}~\cite{shamma2008cooperative,marden2013distributed} --
involves modelling the agents as players in a game and assigning them utility
functions that maximize the efficiency of the game's equilibria.
After agents' utilities are coupled with learning dynamics capable of driving the system to an equilibrium,
an efficient distributed coordination algorithm emerges.
This approach has been utilized in a variety of relevant contexts, including 
collaborative sensing in distributed camera networks~\cite{ding2012collaborative}, 
the distributed control of smart grid nodes~\cite{saad2012game}, 
autonomous vehicle-target assignment~\cite{arslan2007autonomous}
and optimal taxation on road-traffic networks~\cite{paccagnan2019incentivizing}.
A significant advantage of such an approach is that the design of the 
agents' learning dynamics and of the underlying utility structure can be decoupled.
As efficient distributed learning dynamics that drive the agents to an equilibrium 
of the game are already known (see, e.g, \cite{hart2000simple}),
we focus our attention on the design of agents' utility functions in order to 
maximize the efficiency of the equilibria.

The most commonly studied metric in the literature on utility design is the \emph{price of anarchy}~\cite{koutsoupias2009worst},
which is defined as the worst case ratio between the performance at an equilibrium and the best achievable system performance.
Note that a price of anarchy guarantee obtained for a set of utility functions
translates directly to an approximation ratio of the final distributed algorithm.
The majority of the literature focuses primarily on characterizing
the price of anarchy for a given set of player utility functions~\cite{marden2013distributed,marden2014generalized,vetta2002nash},
whereas fewer works design player utilities in order to optimize the price of anarchy~\cite{gairing2009covering,paccagnan2019utility1,chandan2019optimal}.
While several works provide tight bounds on the approximation ratio of polynomial-time centralized algorithms 
for the class of problems we consider (see, e.g., \cite{barman2020tight,sviridenko2017optimal,feige1998threshold}), 
there is currently no result in the literature that establishes comparable bounds 
on the best achievable price of anarchy, 
aside from the general bound put forward in Vetta~\cite{vetta2002nash} that is provably inexact.

\subsection{Model}

In this paper, we consider a class of resource allocation problems
with a set of agents $N = \{1,\dots,n\}$ and a set of resources $\mc{R}$.
Each resource $r \in \mc{R}$ has a corresponding welfare function $W_r:\naturals \to \reals$.
Each agent $i \in N$ must select an action $a_i$ from a corresponding
set of actions $\mc{A}_i \subseteq 2^\mc{R}$.
The system performance under an allocation of agents 
$a = (a_1, \dots, a_n) \in \mc{A} = \mc{A}_1 \times \dots \times \mc{A}_n$
is measured by a function $W : \mc{A} \to \reals_{> 0}$.
The goal is to find an allocation $\opt{a} \in \mc{A}$ that maximizes 
the function 
\begin{equation} \label{eq:welfare_definition}
    W(a) := \sum_{r \in \cup_i a_i} W_r(|a|_r),
\end{equation}
where $|a|_r = |\{ i \in N \text{ s.t. } r \in a_i \}|$ denotes the number of agents
selecting the resource $r$ in allocation $a$.
In this work, we consider nonnegative, nondecreasing concave welfare functions,
i.e., functions that satisfy the following properties:
(i)~$W_r(x)$ is nondecreasing and concave for $x\geq 0$; 
and, (ii)~$W_r(0)=0$ and $W_r(x)>0$ for all $x\geq 1$.
This setup has been thoroughly studied in the submodular maximization 
and game theoretic literature (see, e.g., \cite{arslan2007autonomous,barman2020tight2,barman2020tight,gairing2009covering,sviridenko2017optimal})
as demonstrated by the following two examples:


\begin{example}[General covering problems]
\label{ex:covering_problems}
Consider the general covering problem~\cite{gairing2009covering},
which is a generalization of the max-n-cover problem~\cite{feige1998threshold,hochbaum1996approximating}.
In this setting, we are given a set of elements $E$ and 
$n$ collections $S_1,\dots,S_n$ of subsets of $E$, i.e.,
$S_i \subseteq 2^E$ for all $i=1,\dots,n$.
Each element $e \in E$ has weight $w_e \geq 0$.
The objective is to choose one subset $s_i$ from each collection $S_i$
such that the union $\cup_i s_i$ has maximum total weight, i.e.,
$\sum_{e \in \cup_i s_i} w_e$ is maximized.
We observe that this problem corresponds to a resource allocation problem
where each agent $i\in N$ has action set $\mc{A}_i \subseteq 2^E$,
the action $a_i$ of each agent $i\in N$ corresponds to the subset $s_i$,
and the welfare functions are $W_e(x)=w_e$ for all $e \in E$.
\end{example}


\begin{example}[Vehicle-target assignment problem]
\label{ex:vehicle_target}
Consider the \emph{vehicle-target assignment problem},
first introduced in Murphey~\cite{murphey2000target}, 
and studied by, e.g., Arslan~et~al.~\cite{arslan2007autonomous}
and Barman~et~al.~\cite{barman2020tight2}.
In this setting, we are given a set of $n$ vehicles $N$
and a set of targets $\mc{T}$, where each target $t \in \mc{T}$
has an associated value $v_t > 0$.
Each vehicle $i \in N$ has a set of feasible target assignments 
$\mc{A}_i \subseteq 2^{\mc{T}}$.
Given that a vehicle $i \in N$ is assigned to target $t\in\mc{T}$,
the probability that $t$ is destroyed by $i$ is $p_t \in (0,1]$.
The objective is to compute a joint assignment of vehicles $a \in \Pi_i \mc{A}_i$
that maximizes the expected value of targets destroyed,
which is measured as
\begin{equation}
    W(a) = \sum_{t \in \mc{T}} v_t \cdot (1- (1-p_t)^{|a|_t}),
\end{equation}
where $1-(1-p_t)^x$ is the probability that target $t$ is destroyed
when $x$ vehicles are assigned to it.
Observe that the vehicle-target assignment problem is a 
resource allocation problem 
with nonnegative, nondecreasing concave welfare functions
where the agents are the vehicles, 
the resources are the targets, and the welfare function on 
each resource $t\in\mc{T}$ is $W_t(x)=v_t\cdot(1-(1-p_t)^x)$.\footnote{
    Observe that a vehicle-target assignment problem with $p_t=p=1.0$ 
    for all targets $t\in\mc{T}$ is equivalent to a general covering 
    problem where the chosen subsets correspond with the vehicle assignments
    and the element weights are equal to the target values
    Nonetheless, we retain Example~\ref{ex:covering_problems} as it is
    thoroughly connected to the literature.
}
\end{example}


The focus of this work is on computing near-optimal distributed solutions 
within the class of resource allocation problems described above using 
the game theoretic approach.
To model this particular class of problems, we adopt the framework of 
\emph{resource allocation games}.
A resource allocation game $G=(N,\mc{R},\mc{A},\{F_r\}_{r\in\mc{R}})$
consists of a player set $N$ where each player $i \in N$ evaluates the allocation $a \in \mc{A}$
using a utility function 
\begin{equation} \label{eq:utility_definition}
    U_i(a) := \sum_{r \in a_i} F_r(|a|_r).
\end{equation}
where $F_r : \naturals \to \reals$ defines the utility a player receives at resource $r$
as a function of the total number of agents selecting $r$ in allocation $a$.
We refer to the functions $\{F_r\}_{r\in\mc{R}}$ as the \emph{local utility functions} of the game.
For a given set of welfare functions $\mc{W}$, it is convenient to define a utility mapping 
$\mc{F}:\mc{W}\to\reals^\naturals$, where it is understood that a resource $r$ with welfare 
function $W_r\in\mc{W}$ is assigned the local utility function $\mc{F}(W_r)$.

In the forthcoming analysis, we consider the solution concept of pure Nash equilibrium,
which is defined as any allocation $\nash{a} \in \mc{A}$ such that
\begin{equation} \label{eq:nash_condition}
    U_i(\nash{a}) \geq U_i(a_i, \nash{a}_{-i}), \quad \forall a_i \in \mc{A}_i, \quad \forall i\in N,
\end{equation}
where $a_{-i} = (a_1,\dots,a_{i-1},a_{i+1},\dots,a_n)$.
For a given game $G$, let $\rm{NE}(G)$ denote the set of all allocations $a \in \mc{A}$
that satisfy the Nash condition in Equation~(\ref{eq:nash_condition}).
We define the \emph{price of anarchy} of a game $G$ as\footnote{
    Note that the price of anarchy is well-defined for resource allocation games,
    since these games possess a potential function and, thus, at least one pure Nash equilibrium.
}
\begin{equation}
    \poa(G) := \frac{\min_{a\in\rm{NE}(G)}W(a)}{\max_{a\in\mc{A}}W(a)} \leq 1.
\end{equation}
For a given game $G$, the price of anarchy is the ratio between the system-wide performance of
the worst performing pure Nash equilibrium and the optimal allocation.
The price of anarchy as defined here also applies to the efficiency of the game's
coarse-correlated equilibria~\cite{roughgarden2015intrinsic,chandan2019optimal},
for which many efficient algorithms exist (see, e.g., \cite{hart2000simple}).
We extend the definition of price of anarchy to a given set of games $\mc{G}$, 
which may contain infinitely many instances, as $\poa(\mc{G}) := \inf_{G \in \mc{G}} \poa(G) \leq 1$.
It is important to note that a higher price of anarchy corresponds to an overall improvement 
in the performance of all pure Nash equilibria,
and that $\poa(\mc{G})=1$ implies that all pure Nash equilibria in all games $G\in\mc{G}$ are optimal.
For a given utility mechanism $\mc{F}$, we use the terminology 
``the set of games $\mc{G}$ \emph{induced by} the set of welfare functions $\mc{W}$''
to refer to the set of all games with $W_r\in \mc{W}$ and $F_r=\mc{F}(W_r)$ for all $r\in\mc{R}$.
Given a set $\mc{W}$, our aim is to develop an efficient technique for computing a utility mechanism $\opt{\mc{F}}$ 
that maximizes the price of anarchy in the corresponding set of games $\mc{G}$ induced by $\mc{W}$, 
i.e., we wish to solve
\begin{equation}
    \opt{\mc{F}} \in \argmax_{\mc{F}} \poa(\mc{G}).
\end{equation}

\subsection{Results and Discussion}
\label{section:results}

Our main result is an efficient technique for computing a utility mechanism
that guarantees a price of anarchy of $1-c/e$ in all resource allocation
games with nonnegative, nondecreasing concave welfare functions 
with maximum curvature $c$.

\begin{definition}[Curvature~\cite{conforti1984submodular}] \label{def:curvature}
    The curvature of a nondecreasing concave function $W:\naturals\to\reals$ is
    \begin{equation}
        c = 1-\frac{W(n)-W(n-1)}{W(1)}.
    \end{equation}
\end{definition}

\noindent
In the literature on submodular maximization,
the curvature is commonly used 
to compactly parameterize broad classes of functions.
The notion of curvature we consider was originally defined by Conforti~et~al.~\cite{conforti1984submodular}
in the context of general nondecreasing submodular set functions.
In our specific setup, this reduces to the expression in Definition~\ref{def:curvature}.
Observe that all nondecreasing concave functions have curvature $c \in [0,1]$.
Thus, $c=1$ can be considered in scenarios where the maximum curvature among 
functions in the set $\mc{W}$ is not known.

\begin{theorem}[Informal]
\label{thm:informal}
Let $\mc{G}$ denote the set of all resource allocation games with nonnegative,
nondecreasing concave welfare functions with maximum curvature $c$.
An optimal utility mechanism achieves $\poa(\mc{G})=1-c/e$ 
and can be computed efficiently.
\end{theorem}

A significant consequence of the main result is a universal
guarantee that the best achievable price of anarchy is always 
greater than $1-1/e \approx 63.2\%$ for resource allocation games 
with nonnegative, nondecreasing concave welfare functions.
Note that since $1-1/e$ is the optimal price of anarchy in 
general covering games (see, e.g., Example~\ref{ex:covering_problems}),
it cannot be further improved without more information
about the underlying set of welfare functions.
Our guarantee improves to $1-c/e$ if the curvature $c$ of the 
underlying set of welfare functions is known.

Observe that the result in Theorem~\ref{thm:informal} also implies that one can
efficiently compute a ``universal'' utility mechanism, in that it would guarantee 
a price of anarchy greater than or equal to $1-1/e$ with respect to any game 
with nonnegative, nondecreasing concave welfare functions.
This follows from the observation that $c\leq 1$ always holds.
Of course, if more information is available about the underlying set of welfare functions 
(e.g., the maximum curvature), then this lower bound can be improved.
In the case where the entire set of welfare functions $\mc{W}$ is known \emph{a priori}
and $|\mc{W}|$ is ``small enough'', then the optimal utility mechanism can be computed using 
existing methodologies (see, e.g., \cite{chandan2019optimal}).\footnote{
    In this case, the optimal utility mechanism can be found as the solution of $|\mc{W}|$
    linear programs with number of constraints that is quadratic in the maximum number of agents $n$,
    and $n+1$ decision variables.
    For this reason, the optimal utility mechanism can only be computed
    for modest values of $|\mc{W}|$ and $n$.
}
Consider the sets represented in Figure~\ref{fig:subsets}.
From our reasoning, it holds that as the size of the 
set of welfare functions considered is reduced, 
the prices of anarchy of the 
corresponding optimal utility mechanisms increase.
The set of games induced by welfares in the green ellipse, 
for example, coincides with the vehicle-target assignment problem,
as described in Example~\ref{ex:vehicle_target},
where $p_t=p\in[0,1]$ for all $t\in\mc{T}$.
Note that the welfare function $W_t$ of each target $t\in\mc{T}$ in this problem 
is nonnegative, nondecreasing concave 
(i.e., the green ellipse is a subset of the dotted red box).
Thus, we can immediately observe that the best achievable price of anarchy 
in the corresponding resource allocation game $G$ satisfies 
\[ \poa(G) \geq 1-\frac{1}{e}, \]
which is achieved by the universal utility mechanism from Theorem~\ref{thm:informal}.
Since there is only a single welfare function in this setting
(ignoring uniform scalings) 
the optimal utility mechanism can be computed 
for a modest number of agents, as aforementioned.

\begin{figure}
\centering
\begin{tikzpicture}
    \draw[red,densely dotted,very thick] (-4,-2) rectangle (4,2);
    \node[right,red] () at (-4,1.75) {All nonnegative, nondecreasing concave functions};
    \draw[] (-3.5,-1.75) rectangle (2.75, 1.5);
    \node[right,text width=6.5cm] () at (-3.5,1) {All nonnegative, nondecreasing concave functions with maximum curvature $c$};
    \draw[green!50!black,mark=diamond] (-1,-0.5) ellipse (2.25 and 1);
    \node[green!50!black,text width=4cm] () at (-0.75,-0.5) {Vehicle-target assignment welfares with $p_t=p$, $\forall t\in\mc{T}$};
\end{tikzpicture}
\caption{The set of games induced by the set of all nonnegative, nondecreasing concave functions 
contains the set of all nonnegative, nondecreasing concave functions with maximum curvature $c$, 
which in turn contains the set of all vehicle-target assignment problems with $p_t=p$.}
\label{fig:subsets}
\end{figure}
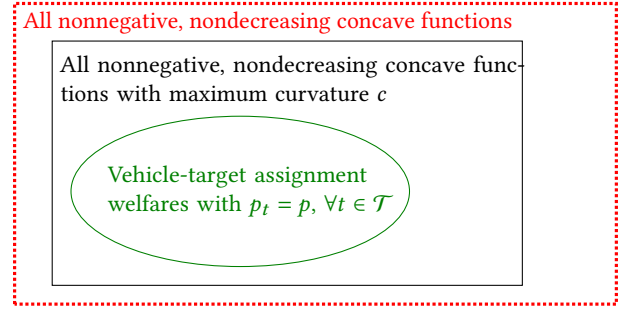

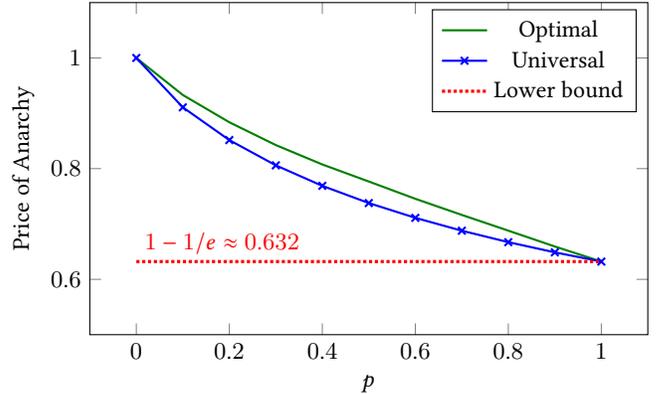
\begin{figure}[t]
\centering
\begin{tikzpicture}
    \begin{axis}[
        height=6cm, width=9cm,
        ymin=0.5, ymax=1.1,
        xlabel=$p$,ylabel={Price of Anarchy},
    ]
    \addplot[
        domain=0:1,
        samples=10,
        color=green!50!black,
        thick,
    ]
    coordinates {
        (0,1) 
        (0.1,0.9330969710206909)
        (0.2,0.8838285124383216)
        (0.30000000000000004,0.842579669620065)
        (0.4,0.807577632411875)
        (0.5,0.7767356461510652)
        (0.6,0.7455275380737876)
        (0.7000000000000001,0.7163904490282815)
        (0.8,0.6879681610634845)
        (0.9,0.6593665875485195)
        (1.0,0.6321205588285508)
    };
    \addlegendentry{Optimal};
    \addplot[
        domain=0:1,
        samples=10,
        color=blue,
        thick, mark=x,
    ]
    coordinates{
        (0,1)
        (0.1,0.9108092426192466)
        (0.2,0.8516066222121267)
        (0.30000000000000004,0.805961935454864)
        (0.4,0.7689021787364665)
        (0.5,0.7378094488011929)
        (0.6,0.7111201368784607)
        (0.7000000000000001,0.6878172273937563)
        (0.8,0.6671988197637683)
        (0.9,0.6487589648755244)
        (1.0,0.6321205548327937)
    };
    \addlegendentry{Universal}
    \addplot[
        domain=0.0001:1,
        samples=100,
        color=red,
        very thick,densely dotted,
    ]
    {1-1/exp(1)};
    \addlegendentry{Lower bound};
    \node[right,red] () at (axis cs:0,0.665) {$1-1/e \approx 0.632$};
    \end{axis}
\end{tikzpicture}
\caption{
The price of anarchy of the universal utility mechanism obtained in this work 
and the optimal utility mechanism in the vehicle-target assignment problems with $p_t\in[0,p]$ for all $t\in\mc{T}$.
Note that this utility mechanism is designed for the set of all nonnegative, nondecreasing
concave welfare functions but its price of anarchy is close to the best achievable
within this particular setting.
}
\label{fig:theoretical_bounds}
\end{figure}

In Figure~\ref{fig:theoretical_bounds}, we plot 
the price of anarchy corresponding to the optimal utility mechanism within this setting (labelled ``Optimal''),
the price of anarchy achieved by the universal utility mechanism (labelled ``Universal'')
and the $1-1/e$ lower bound from Theorem~\ref{thm:informal} (labelled ``Lower bound'').
As expected, the optimal utility mechanism corresponds with the 
best price of anarchy as it was designed specifically for the 
underlying welfare function.
However, knowledge of the set of welfare functions 
corresponds with only a small increase in the price of anarchy;
the price of anarchy achieved by the universal utility mechanism
is surpisingly close to the best achievable by any mechanism for all values of $p\in[0,1]$.
Note that the universal utility mechanism is only guaranteed
to achieve a price of anarchy of $1-1/e$.

Consider once again the sets represented in Figure~\ref{fig:subsets}.
The above example suggests that the utility mechanism designed to maximize the price of anarchy 
in the set of games induced by all welfare functions in the dotted, red box 
may achieve price of anarchy close to the optimal 
within the sets of games induced by any subset of the dotted red box,
as we have observed that this holds for the set of games induced by 
welfare functions in the green ellipse.
While we do not provide formal proofs for these observations,
they provide further motivation for deriving efficient techniques
for computing utility mechanisms that maximize the price of anarchy
with respect to broad classes of welfare functions.

\subsection{Related Works}
Submodular resource allocation problems have been the focus of a
significant research effort for many years, particularly in the optimization community.
Since the computation of an optimal allocation in such problems is $\mc{NP}$-hard in general,
many researchers have focused on providing approximation guarantees for polynomial-time algorithms.
For example, approximate solutions to max-$n$-cover problems were studied by 
Feige~\cite{feige1998threshold} and Hochbaum~\cite{hochbaum1996approximating} almost 25 years ago.
In the latter manuscript, the greedy algorithm is show to have 
an approximation ratio of $1-1/e$.
Recently, Sviridenko~et~al.~\cite{sviridenko2017optimal} proposed a polynomial-time
algorithm for computing approximate solutions that perform within $1-c/e$ of the optimal 
for the class of resource allocation problems with nonnegative, nondecreasing submodular 
welfare functions with curvature $c$. 
Barman~et~al.~\cite{barman2020tight} provide a polynomial-time algorithm
that returns allocations with a $1-k^k e^{-k}/(k!)$ approximation ratio 
in resource allocation problems with welfare functions 
$W_r(x)=\min\{x,k\}$ for $k\in\naturals_{\geq 1}$, for all $r\in\mc{R}$.
In their respective works, all three of the approximation ratios provided above are also shown to be 
the best achievable, i.e., 
it is shown that there exist no other polynomial-time algorithm capable 
of always computing an approximate solution that is closer to the optimal unless $\mc{P}=\mc{NP}$.
In this work, we obtain price of anarchy guarantees in resource allocation games
that match these approximation ratios from submodular maximization.

Although utility mechanisms have been studied in resource allocation games,
the majority of results have focused on deriving price of anarchy bounds for given utility structures
(e.g., marginal contribution, equal shares, etc.)~\cite{marden2013distributed,marden2014generalized}.
In this respect, Vetta~\cite{vetta2002nash} proves that there always exist
player utility functions that guarantee a price of anarchy larger than $50\%$ 
within a more general class of games than those we consider here.\footnote{
    In the class of \emph{valid-utility games}, the system objective $W:\mc{A}\to\reals$ 
    is a nondecreasing submodular set function over the agents' actions
    and is not necessarily separable over a set of resources; much more
    general than the class of resource allocation games with nonnegative,
    nondecreasing concave welfare functions.
}
The notion of utility mechanisms that \emph{maximize} -- or otherwise improve --
the price of anarchy was introduced in Christodoulou~et~al.~\cite{christodoulou2004coordination}.
This approach has been applied to many distributed optimization problems,
including machine scheduling~\cite{caragiannis2013efficient,immorlica2009coordination}, 
selfish routing~\cite{bilo2019dynamic,caragiannis2010taxes} and auction mechanism design~\cite{roughgarden2014barriers,syrgkanis2013composable}.
A prominent example in this line of research is Gairing~\cite{gairing2009covering},
who proves that the best achievable price of anarchy in covering games is $1-1/e$ and derived an
optimal utility mechanism.
We provide an efficient technique for computing a utility mechanism
that achieves a price of anarchy larger than $1-1/e\approx 63.2\%$ 
in all resource allocation games with nonnegative, nondecreasing concave welfare functions, 
which effectively generalizes the result in~\cite{gairing2009covering}
and significantly improves upon the bound provided in~\cite{vetta2002nash}
by exploiting the structure on $W$ (see Equation~(\ref{eq:welfare_definition})).

More recently, Chandan~et~al.~\cite{chandan2019optimal} proposed a linear 
programming based approach for computing a utility mechanism
that maximizes the price of anarchy for a given set of resource allocation games.
Unfortunately, their approach does not
provide \emph{a priori} guarantees on the price of anarchy achieved.
In other words, a concrete lower bound on the best achievable price of anarchy 
using their approach cannot be obtained without first solving 
the linear programs corresponding to the underlying set of welfare functions.
This can only be accomplished with this linear programming based approach
for a finite number of welfare functions and a modest number of players.
Our goal in this work is to derive guarantees on the best achievable 
price of anarchy for games induced by the set of all welfare functions satisfying 
a few specific properties (e.g., concavity, maximum curvature).
In the next section, we provide an explicit expression for a utility mechanism that is guaranteed
to have price of anarchy greater than or equal to $1-c/e$ for all resource allocation games with 
nonnegative, nondecreasing concave welfare functions with maximum curvature $c$.
Since the linear programming based approach returns a utility mechanism that maximizes the price of anarchy, 
the claim in Theorem~\ref{thm:informal} extends to the solutions of the linear program.

\subsection{Organization}

The remainder of the paper is structured as follows:
Section~\ref{section:main_result} presents the proof of the main result
and an extension result for more specific sets of welfare functions.
Section~\ref{section:simulation} showcases our simulation example and accompanying discussion.
Section~\ref{section:conclusion} concludes the manuscript and provides a brief discussion on potential future directions.
All proofs omitted from the manuscript are provided in the appendix, for ease of exposition.

\section{Main Result and Extensions}
\label{section:main_result}

In this section, we prove the claim in Theorem~\ref{thm:informal} by constructing 
a utility mechanism that achieves the best achievable price of anarchy of $1-c/e$ 
with respect to the set of all nonnegative, nondecreasing concave welfare functions 
with maximum curvature $c\in[0,1]$.
In scenarios where a more specific set of welfare functions is considered,
we outline how the techniques used to prove Theorem~\ref{thm:informal}
can be generalized to derive tighter \emph{a priori} bounds on the best 
achievable price of anarchy.

Before presenting the proof of Theorem~\ref{thm:informal}, we provide an informal 
outline of the three steps underpinning the result.
These steps correspond with the three parts of the formal proof, but
are listed in a different order for sake of clarity.
For the reader's convenience, we include the part of the proof that corresponds
with each of the steps in our informal outline.
The proof is summarized as follows:

\vspace{.2cm}\noindent\emph{--Step \#1:} 
We demonstrate that any concave welfare function can be decomposed as a linear 
combination with nonnegative coefficients of a specialized set of basis functions. 
[Section~\ref{sec:main_proof}, Part~ii)]

\vspace{.2cm}\noindent\emph{--Step \#2:} 
We derive optimal basis utility functions for each of the basis functions in the 
specialized set. [Section~\ref{sec:main_proof}, Part~i)]

\vspace{.2cm}\noindent\emph{--Step \#3:} 
We construct local utility functions as linear combinations over the optimal basis utility 
functions from Step 2 with the nonnegative coefficients derived in Step 1. 
Finally, we demonstrate that this tractable approach for constructing resource 
utility functions provides near optimal efficiency guarantees. 
[Section~\ref{sec:main_proof}, Part~iii)]

\subsection{Proof of Theorem~\ref{thm:informal}} \label{sec:main_proof}
Here we consider the class of games induced by the set of all concave welfare functions
with maximum curvature $c\in[0,1]$.
The proof of Theorem~\ref{thm:informal} proceeds in the following three parts:

\begin{enumerate}[leftmargin=*]
    \item[i)] Given a value $c\in[0,1]$, we derive explicit expressions for 
    the local utility functions that maximize the price of anarchy relative to a 
    restricted class of nonnegative, nondecreasing concave welfare functions with curvature $c$.
    Among the optimal price of anarchy values obtained for the functions in this restricted class, 
    the lowest is equal to $1-c/e$;
    \item[ii)] We show that any nonnegative, nondecreasing concave welfare function $W$ 
    with curvature less than or equal to $c$ can be represented 
    as a linear combination with explicitly defined nonnegative coefficients over this restricted class; and,
    \item[iii)] We demonstrate that using the local utility functions computed as a linear combination
    over the optimal local utility functions from i) with the nonnegative coefficients from ii)
    guarantees that $\poa(\mc{G})=1-c/e$ within the set of resource allocation games $\mc{G}$
    induced by all nonnegative, nondecreasing concave welfare functions with maximum curvature $c$.
\end{enumerate}

The above parts successfully prove Theorem~\ref{thm:informal} as we argue here.
Note that, 
by part i), the lowest optimal price of anarchy among welfare functions 
in the restricted class considered is equal to $1-c/e$, for given curvature $c\in[0,1]$.
By part iii), this implies that all resource allocation games induced by 
nonnegative, nondecreasing concave welfare functions with 
maximum curvature $c$ have optimal price of anarchy equal to $1-c/e$.
This is because, by part ii), any such welfare function can be represented as 
a nonnegative linear combination over the restricted class of welfare functions we consider.
Since the best achievable price of anarchy for at least one of the functions in the restricted class is also $1-c/e$,
one cannot further improve the price of anarchy within the set of games considered.
In addition, parts i)--iii) combine to prove that a corresponding utility mechanism
that maximizes the price of anarchy entails computing nonnegative linear combinations over a class of functions 
with explicit expressions.
Thus, the computation of optimal local utility functions 
is polynomial in the number of players.


\vspace*{.25cm}\noindent{\bf Part i).} 
In this part of the proof, we provide explicit expressions for local utility functions
that maximize the price of anarchy with respect to a restricted set of welfare functions, 
as well as the corresponding optimal price of anarchy.
To that end, given parameters $\alpha\in[0,1]$ and $\beta\in\naturals_{\geq 1}$,
we define the \emph{$(\alpha,\beta)$-coverage} function as
\begin{equation}
    V^\alpha_\beta(x) := (1-\alpha)\cdot x + \alpha\cdot\min\{x,\beta\}.
\end{equation}
It is straightforward to verify that every $(\alpha,\beta)$-coverage function is 
nonnegative, nondecreasing concave.
In the lemma below, we derive a local utility function that maximizes 
the price of anarchy of the set of resource allocation games induced by any
given $(\alpha,\beta)$-coverage function.
We use this result to derive the 
optimal utility functions for a broad range of local welfare functions
in Part iii).

\begin{lemma} \label{thm:coverage_solution}
Consider the set of resource allocation games $\mc{G}$ induced by 
the $(\alpha,\beta)$-coverage function
\[ V^\alpha_\beta(x) = (1-\alpha)\cdot x + \alpha\cdot\min\{x,\beta\}, \]
where $\alpha \in [0,1]$ and $\beta \in \naturals_{\geq 1}$.
Let $\rho = (1 - \alpha \cdot \beta^\beta e^{-\beta} / (\beta!) )^{-1}$,
and define $F^\alpha_\beta$ as in the following recursion: $F^\alpha_\beta(1):=W(1)$,
\begin{equation}
    F^\alpha_\beta(x+1):=\max\Big\{ \frac{1}{\beta} [x F^\alpha_\beta(x) - V^\alpha_\beta(x) \rho] + 1, 1-\alpha \Big\}, \> \forall x=1,\dots,n-1.
\end{equation}
Then, the local utility function $F^\alpha_\beta$ maximizes the price of anarchy and 
the corresponding price of anarchy is $\poa(\mc{G})=1/\rho$.
\end{lemma}

\begin{proof}
    The proof is presented in Appendix~A.1.
\end{proof}

According to the result in Lemma~\ref{thm:coverage_solution},
the maximum achievable price of anarchy in resource allocation games 
induced by a $(\alpha,\beta)$-coverage function with $\alpha=1$ and $\beta\geq 1$ is $1- \beta^\beta e^{-\beta} / (\beta!)$.
Surprisingly, Barman~et~al.~\cite{barman2020tight} show that that the optimal approximation ratio 
of any polynomial-time algorithm for the same class of resource allocation problems is also 
$1- \beta^\beta e^{-\beta} / (\beta!)$.
Similarly, the optimal price of anarchy for the $(\alpha,\beta)$-coverage function 
with $\alpha\in[0,1]$ and $\beta=1$ is $1-\alpha/e$,
which matches the best achievable approximation ratio of any polynomial-time algorithm
for this problem setting~\cite{sviridenko2017optimal}.


\vspace*{.25cm}\noindent{\bf Part ii).}
In the next result, we show that any nonnegative, nondecreasing concave welfare function
with maximum curvature $c\in[0,1]$ can be represented as a nonnegative linear combination 
over the set of $(c,k)$-coverage functions with $k=1,\dots,n$.

\begin{lemma}
\label{lem:nonnegative_lincomb}
Let $W:\naturals\to\reals$ denote a nonnegative, nondecreasing concave function
with curvature less than or equal to $c\in[0,1]$.
Then, the nonnegative coefficients $\eta_1,\dots,\eta_n$ satisfy
\begin{equation}
    W(x) = \sum^n_{k=1} \eta_k \cdot V^c_k(x), \quad \forall x=0,1,\dots,n,
\end{equation}
where $\eta_1 := [ 2 W(1) - W(2) ]/c$,
$\eta_k := [ 2 W(k) - W(k-1) - W(k+1) ]/c$, for $k=2,\dots,n-1$,
and $\eta_n := W(1) - \sum^{n-1}_{k=1} \eta_k$.
\end{lemma}

\begin{proof}
    The proof is presented in Appendix~A.3.
\end{proof}


\vspace*{.25cm}\noindent{\bf Part iii).}
We begin by describing a utility mechanism parameterized by the maximum curvature
and maximum number of players.
Let $\mc{G}$ denote the set of resource allocation games
induced by all nonnegative, nondecreasing concave 
functions with maximum curvature $c\in[0,1]$
with a maximum of $n$ players .
Consider any resource allocation game $G \in \mc{G}$
and assign the following local utility function to each $r\in\mc{R}$: 
\[ F_r(x) = \sum^n_{k=1} \eta_k \cdot F^c_k(x), \quad \forall x=1,\dots,n, \]
where $\eta_1 := [ 2 W_r(1) - W_r(2) ]/c$,
$\eta_k := [ 2 W_r(k) - W_r(k-1) - W_r(k+1) ]/c$, for $k=2,\dots,n-1$,
and $\eta_n := W_r(1) - \sum^{n-1}_{k=1} \eta_k$,
$W_r:\naturals\to\reals$ is the welfare function on the resource $r$
and each function $F^c_k:\naturals\to\reals$, $k=1,\dots,n$, is the optimal local utility function for $V^c_k(x)$ 
defined recursively in Lemma~\ref{thm:coverage_solution}.
In this part, we show that $\poa(G)\geq 1-c/e$ holds
for this utility mechanism.

Given maximum curvature $c\in[0,1]$, Lemma~\ref{thm:coverage_solution} proves that among the $(c,k)$-coverage functions 
with $k=1,\dots,n$, the $(c,1)$-coverage function has best achievable price of anarchy $1-c/e$ which is strictly 
lower than the best achievable price of anarchy for any $(c,k)$-coverage function with $k>1$.
This implies that the best achievable price of anarchy must satisfy $\poa(\mc{G}) \leq 1-c/e$, 
since any game $G$ in the set of resource allocation games induced by the $(c,1)$-coverage function must also be 
in the set $\mc{G}$, i.e., $G \in \mc{G}$, and there is at least one such game with $\poa(G)=1-c/e$.
We now show that $\poa(\mc{G})\geq 1-c/e$ also holds.
Recall from Lemma~\ref{lem:nonnegative_lincomb} that 
the nonnegative coefficients $\eta_1,\dots,\eta_n$ defined above satisfy 
\[ W_r(x) = \sum^n_{k=1} \eta_k \cdot V^c_k(x) \quad \forall x=0,1,\dots,n. \]
It must then hold that, for any $r\in\mc{R}$, $(F_r, (1-c/e)^{-1})$ is a feasible point 
in the linear program in Equation~(12) (see Appendix~A)
for any $n$ and the corresponding $W_r$.
Observe that each constraint in the linear program must be satisfied since, 
by Lemma~\ref{lem:nonnegative_lincomb},
it can be represented as a nonnegative linear combination of the constraints in the 
$n$ linear programs for $V^c_k$ and $(F^c_k, (1-c/e)^{-1})$, $k=1,\dots,n$, i.e.,
for all $r\in\mc{R}$ and all $(x,y,z) \in \mc{I}(n)$ it must hold that
\begin{equation*}
\begin{aligned}
    (1-c/e&)^{-1} W_r(x) \geq \sum^n_{k=1} \eta_k \cdot \left[ 1-c\cdot \frac{k^k e^{-k}}{k!} \right] V^c_k(x) \\
                 \geq\> &\sum^n_{k=1} \eta_k \cdot \left[ V^c_k(y) + (x-z) F^c_k(x) - (y-z) F^c_k(x+1) \right] \\
                 =\> &W_r(y) + (x-z) F_r(x) - (y-z) F_r(x+1),
\end{aligned}
\end{equation*}
where the first inequality holds because $1-c/e \leq 1-c\cdot k^k e^{-k}/(k!)$ for all $k\geq 1$ and
since $W_r$, $V^c_k(x)$, $k=1,\dots,n$, and the coefficients $\eta_1,\dots,\eta_n$ are nonnegative,
and the second inequality holds by the result in Lemma~\ref{thm:coverage_solution}.


\subsection{Specialized sets of welfare functions}
In the previous subsection, we used a series of arguments
to prove the bound on the price of anarchy in Theorem~\ref{thm:informal}.
Informally, we considered a specified set of candidate welfare functions.
For this set of candidate welfare functions, we derived a corresponding 
set of local utility functions that maximize the price of anarchy.
Finally, we showed that the best achievable price of anarchy for these candidates
is automatically a lower bound on the best achievable price of anarchy 
across a much broader set of welfare functions.
A set of candidate welfare functions must be chosen for two reasons:
(i)~an optimal local utility function and its corresponding optimal price of anarchy 
can be obtained in advance for each of the candidate welfare functions; and, more importantly,
(ii)~any function within the set of welfare functions of interest can be expressed 
as a nonnegative linear combination over the set of candidate welfare functions,
thus inheriting the same optimal price of anarchy.
Clearly the choice of candidate functions is important, 
as the \emph{a priori} guarantees on the 
price of anarchy is characterized by the best achievable price of anarchy 
corresponding to each candidate.

As our next result, we outline a mechanism for obtaining a 
set of candidate functions for a given set of welfare functions $\mc{W}$
such that any function $W\in\mc{W}$ can be expressed as a nonnegative
linear combination over the candidate functions.
This generalizes the approach taken in the previous subsection
to sets of resource allocation games for which more is known about the welfare
functions than concavity and maximum curvature $c\in[0,1]$.

\begin{corollary}
\label{cor:generalized_candidate}
Let $\mc{W}$ denote a set of nonnegative, nondecreasing concave welfare functions 
and $n$ be the maximum number of agents.
Let $W^{\rm ub}$ and $W^{\rm lb}$ be two nonnegative, nondecreasing concave 
functions that satisfy the following for all $W\in\mc{W}$:
(i)~$W^{lb}(x+1)-W^{lb}(x) \leq [W(x+1)-W(x)]/W(1) \leq W^{ub}(x+1)-W^{ub}(x)$, 
for all $x=1,\dots,n-1$; and,
(ii)~$[W(x+1)-2W(x)+W(x-1)]/W(1) \leq W^{ub}(x+1)-2W^{ub}(x)+W^{ub}(x-1) \leq W^{lb}(x+1)-2W^{lb}(x)+W^{lb}(x-1)$,
for all $x=2,\dots,n-1$.
Finally, define the candidate functions $W^{(k)}$, $k=1,\dots,n$, as follows:
\begin{equation}
    W^{(k)}(x) = \begin{cases}
        W^{ub}(x)                         \quad &\text{if } 1\leq x\leq k, \\
        W^{ub}(k) + W^{lb}(x) - W^{lb}(k) \quad &\text{if } x>k.
    \end{cases}
\end{equation}

Then, for any welfare function $W\in\mc{W}$,
there exist nonnegative coefficients $\eta_1,\dots,\eta_n$ that satisfy
\[ W(x) = \sum^n_{k=1} \eta_k \cdot W^{(k)}(x), \quad \forall x=0,1,\dots,n. \]
\end{corollary}

\begin{proof}
    The proof is in Appendix~B.
\end{proof}

We highlight several important implications of the result in Corollary~\ref{cor:generalized_candidate} 
in the following discussion: 

\vspace*{.25cm}\noindent (i)~We showed in Part iii) of the previous subsection that any set 
of resource allocation games $\mc{G}$ induced by nonnegative linear combinations over a set of candidate 
functions $W^{(1)},\dots,W^{(n)}$ automatically inherits the optimal price of anarchy guarantees 
of the candidates, i.e., there exist local utility functions such that $\poa(\mc{G})$ is greater than 
or equal to the lowest optimal price of anarchy among the candidates.
Thus, by simply precomputing the optimal local utility functions $F^{(1)},\dots,F^{(n)}$ 
and price of anarchy bounds corresponding to the candidate functions,
one obtains a lower bound on the best achievable price of anarchy in the set of games considered.
This can be done, for example, 
using the linear programming based methodology proposed in~\cite{chandan2019optimal}.

\vspace*{.25cm}\noindent (ii)~If the candidate function with lowest corresponding optimal 
price of anarchy happens to be a member of the underlying set $\mc{W}$,
then we can also say that this lower bound is the best achievable price of anarchy.
Furthermore, an optimal utility mechanism then consists of computing nonnegative linear 
combination over the precomputed functions $F^{(1)},\dots,F^{(n)}$.

\vspace*{.25cm}\noindent (iii)~The complexity of computing the local utility functions that 
achieve the lower bound on $\poa(\mc{G})$ is polynomial in the number of players.
This follows from observing that the functions $F^{(1)},\dots,F^{(n)}$ can be precomputed 
and there is a closed-form expression for the nonnegative coefficients $\eta_k$, $k=1,\dots,n$, 
given a welfare function $W \in \mc{W}$ (see, e.g., the proof of Corollary~\ref{cor:generalized_candidate}).

\section{Simulation Results}
\label{section:simulation}

In this section, we provide
an in-depth simulation example
in which we compare the equilibrium performance
corresponding to the universal utility mechanism we derive in the previous section for $c=1$
against two well-studied utility structures from the literature:
the \emph{identical interest utility} and the \emph{equal shares utility mechanism}.
The identical interest utility precisely aligns
the players' utilities to the system objective, 
i.e., $U_i(a)=W(a)$ for all $i\in N$.
Observe that under this utility, if $U_i(a_i,a_{-i}) > U_i(a'_i,a_{-i})$ for a player $i\in N$,
then it must hold that $W(a_i,a_{-i}) > W(a'_i,a_{-i})$.
As its name suggests, the equal shares utility mechanism distributes the welfare obtained 
on each resource among the players selecting that resource
which corresponds with local utility functions of the form $F^{\rm es}_r(x)=W_r(x)/x$ for all $r\in\mc{R}$.
At first glance, one might expect that one of these two 
utilities would be best, e.g.,
the identical interest utility exposes the players 
to the actual system objective.
However, in terms of the worst-case equilibrium efficiency,
our simulation provides concrete evidence that 
the universal utility mechanism performs better.

Consider a vehicle-target assignment problem with $n=10$ vehicles and $|\mc{T}|=n+1$ targets,
where $\mc{T}=\{t_1,\dots,t_{n+1}\}$. 
We purposely choose a small number of vehicles (i.e., $n=10$) in order to allow for explicit computation of the 
optimal allocation and, therefore, of the corresponding price of anarchy.
Each vehicle $i\in N$ has two singleton target assignments chosen randomly from a uniform distribution 
over the $n+1$ targets, i.e., $\mc{A}_i=\{\{t_j\},\{t_k\}\}$ where $j,k \sim \mc{U}\{1,n+1\}$.
Each target $t\in\mc{T}$ has welfare function $W_t(x)=v_t \cdot (1-(1-p)^x)$
where $v_t$ is drawn from a uniform distribution over the interval $[0,1]$ and $p\in[0,1]$
is a given parameter.

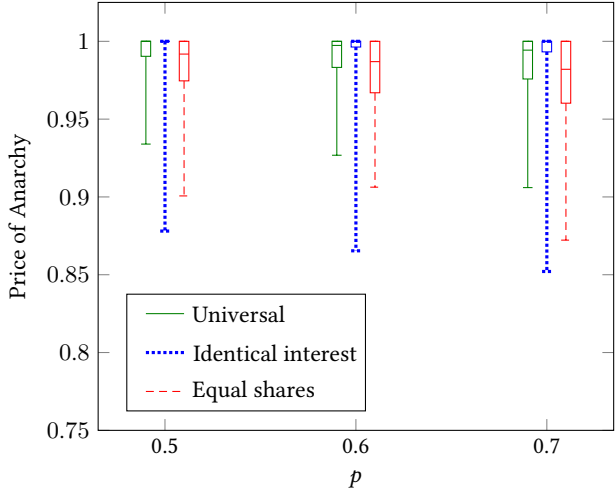
\begin{figure}[t]
\begin{tikzpicture}
    \begin{axis}[
        boxplot/draw direction=y,
        xtick={5,15,25},
        xticklabels={0.5,0.6,0.7},
        xlabel={$p$},
        ylabel={Price of Anarchy},
        ymin=0.75,
    ]
        \draw[] (axis cs:3,0.7625) rectangle (axis cs:15.5,0.8375);
        \draw[green!50!black]                   (axis cs:4,0.825) -- (6,0.825) node[right,black]{Universal};
        \draw[blue,densely dotted,very thick]   (axis cs:4,0.8)   -- (6,0.8)   node[right,black]{Identical interest};
        \draw[red,densely dashed]               (axis cs:4,0.775) -- (6,0.775) node[right,black]{Equal shares};
        \addplot+[
            boxplot prepared={
                draw position=4,
                median=1.,
                upper quartile=1.,
                lower quartile=0.99035963,
                upper whisker=1.,
                lower whisker=0.93396285,
                box extend=0.5,  
                whisker extend=0.5, 
                every box/.style={draw=green!50!black},
                every whisker/.style={green!50!black},
                every median/.style={green!50!black},
            }
        ] coordinates {};
        \addplot+[
            boxplot prepared={
                draw position=14,
                median=0.99736221,
                upper quartile=1.,
                lower quartile=0.98323658,
                upper whisker=1.,
                lower whisker=0.92678849,
                box extend=0.5,  
                whisker extend=0.5, 
                every box/.style={draw=green!50!black},
                every whisker/.style={green!50!black},
                every median/.style={green!50!black},
            }
        ] coordinates {};
        \addplot+[
            boxplot prepared={
                draw position=24,
                median=0.9943302,
                upper quartile=1.,
                lower quartile=0.97579496,
                upper whisker=1.,
                lower whisker=0.90599849,
                box extend=0.5,  
                whisker extend=0.5, 
                every box/.style={draw=green!50!black},
                every whisker/.style={green!50!black},
                every median/.style={green!50!black},
            }
        ] coordinates {};
        \addplot+[
            boxplot prepared={
                draw position=5,
                median=1.,
                upper quartile=1.,
                lower quartile=1.,
                upper whisker=1.,
                lower whisker=0.87806205,
                box extend=0.5,  
                whisker extend=0.5, 
                every box/.style={draw=blue},
                every whisker/.style={very thick,densely dotted,blue},
                every median/.style={blue},
            }
        ] coordinates {};
        \addplot+[
            boxplot prepared={
                draw position=15,
                median=1.,
                upper quartile=1.,
                lower quartile=0.99632867,
                upper whisker=1.,
                lower whisker=0.8653793,
                box extend=0.5,  
                whisker extend=0.5, 
                every box/.style={draw=blue},
                every whisker/.style={very thick,densely dotted,blue},
                every median/.style={blue},
            }
        ] coordinates {};
        \addplot+[
            boxplot prepared={
                draw position=25,
                median=1.,                  
                upper quartile=1.,          
                lower quartile=0.99322341,  
                upper whisker=1.,           
                lower whisker=0.85209072,   
                box extend=0.5,  
                whisker extend=0.5, 
                every box/.style={draw=blue,solid},
                every whisker/.style={very thick,densely dotted,blue},
                every median/.style={solid,blue},
            }
        ] coordinates {};
        \addplot+[
            boxplot prepared={
                draw position=6,
                median=0.99180013,
                upper quartile=1.,
                lower quartile=0.97458916,
                upper whisker=1.,
                lower whisker=0.90065523,
                box extend=0.5,  
                whisker extend=0.5, 
                every box/.style={draw=red,solid},
                every whisker/.style={densely dashed,red},
                every median/.style={solid,red},
            }
        ] coordinates {};
        \addplot+[
            boxplot prepared={
                draw position=16,
                median=0.98694877,
                upper quartile=1.,
                lower quartile=0.96693901,
                upper whisker=1.,
                lower whisker=0.90627167,
                box extend=0.5,  
                whisker extend=0.5, 
                every box/.style={draw=red,solid},
                every whisker/.style={densely dashed,red},
                every median/.style={solid,red},
            }
        ] coordinates {};
        \addplot+[
            boxplot prepared={
                draw position=26,
                median=0.98202759,                  
                upper quartile=1.,          
                lower quartile=0.96023658,  
                upper whisker=1.,           
                lower whisker=0.87223827,   
                box extend=0.5,  
                whisker extend=0.5, 
                every box/.style={draw=red,solid},
                every whisker/.style={densely dashed,red},
                every median/.style={solid,red},
            }
        ] coordinates {};
    \end{axis}
\end{tikzpicture}
\caption{Box plots depicting the equilibrium efficiency measured across $T=10^3$ instances 
for the universal utility, identical interest utility and equal shares utility mechanisms in
the vehicle-target assignment problem with $p_t=p$ for all $t\in\mc{T}$ and $p\in\{0.5,0.6,0.7\}$.
Note that among the three utility mechanisms studied, 
the price of anarchy is highest for the universal utility mechanism.}
\label{fig:simulation}
\end{figure}

Within the scenario described above, we model agent decision making as best response dynamics over $T=100$ iterations.
More specifically, the agents best respond in a round robin fashion to the actions of the others, i.e.,
at each time step $t\in\{1,\dots,T\}$, the agent $i=t \mod n$ selects an action $a^t_i \in \mc{A}_i$
such that $U_i(a^t_i,a^{t-1}_{-i})=\max_{a_i \in \mc{A}_i} U_i(a_i,a^{t-1}_{-i})$, and then
$a^t=(a^t_i,a^{t-1}_{-i})$.
As the agents settled to a pure Nash equilibrium within $20$ iterations in all the instances 
we generated, repeating over $T=100$ iterations is justified.
We ran our simulations for the three utility structures described 
(i.e., universal utility mechanism, identical interest utility and equal shares utility mechanism) 
over $10^3$ randomly generated instances for $p\in\{0.5,0.6,0.7\}$, as described above.
The price of anarchy data was obtained by dividing the welfare at equilibrium by the 
best achievable welfare computed by exhaustive search.
The box plots in Figure~\ref{fig:simulation} display statistics on the price of anarchy values we obtain 
in our simulations.
These box plots are to be interpreted as follows:
(i)~the top and bottom of the boxes correspond to the 75-th and 25-th percentiles of the price of anarchy,
respectively;
(ii)~the top and bottom ``whiskers'' show the maximum and minimum price of anarchy, respectively; and,
(iii)~each of the boxes is bisected by the median value of the corresponding prices of anarchy.

Observe that for all three values of $p$ considered, the minimum price of anarchy across the $10^3$ randomly generated 
instances is highest for the universal utility mechanism, as expected.
However, for all three utility functions considered, the maximum and 75-th percentile
of the price of anarchy data collected is always at $1$, i.e.,
the best response dynamics settled on an optimal allocation for at least 25\% of the randomly generated instances.  
In fact, all of the other statistics on the price of anarchy are skewed away from the minimum,
suggesting that the worst-case instances are quite rare.
Furthermore, although the minimum price of anarchy for the identical interest utility is lowest for $p\in\{0.5,0.6,0.7\}$,
the identical interest utility also has the highest median and 25-th percentile price of anarchy values
among the three utilities considered.
These observations suggest that -- as one might expect -- the price of anarchy is not representative of the average equilibrium efficiency,
and that the identical interest utility could perform better than the universal utility mechanism in this respect.
The design of utility functions that maximize the expected equilibrium efficiency
could be a fruitful direction for future work.

\section{Conclusions and Open Questions}
\label{section:conclusion}

In this work, we consider the game theoretic approach 
to the design of distributed algorithms for resource allocation
problems with nonnegative, nondecreasing concave welfare functions.
Our main result is that there exist utility mechanisms 
that achieve a price of anarchy $1-c/e$ in resource allocation games
with nonnegative, nondecreasing concave welfare functions with
maximum curvature $c\in[0,1]$.
In cases where the maximum curvature is not known, 
the guarantee corresponding to $c=1$ still applies.
Furthermore, we show that the local utility functions can be
computed in polynomial time as nonnegative linear combination
over a restricted set of functions with explicit expressions.

In the example we studied in Section~\ref{section:results},
we observed that the price of anarchy achieved by the universal utility mechanism is near-optimal
within sets of games induced by specialized welfare sets.
Considering the gains in tractability and generality when using this mechanism,
this small decrease in equilibrium efficiency guarantees may be acceptable.
Future work should characterize the difference between the
price of anarchy achieved by the universal utility mechanism 
and the best achievable price of anarchy within the set of games 
induced by a given set of welfare functions.

We observed that, in certain cases, the price of anarchy guarantees that we obtain match 
the best-achievable approximation ratios among polynomial-time centralized 
algorithms~\cite{barman2020tight,sviridenko2017optimal}.
An investigation into the potential connections between
the best achievable price of anarchy in resource allocation games
and the best achievable approximation ratio among polynomial-time centralized algorithms
would reflect on the relative performance of distributed and centralized
multiagent coordination algorithms.

Since the price of anarchy is a measure for the worst-case equilibrium efficiency within a family of instances,
it may not be representative of the expected performance of a distributed algorithm designed using the game theoretic approach.
This is demonstrated, for example, by the simulation results studied in Section~\ref{section:simulation}.
A relevant research direction is the design of player utility functions with the objective 
of maximizing the expected equilibrium efficiency.




\begin{acks}
This work is supported by \texttt{ONR Grant \#N00014-20-1-2359}
and \texttt{AFOSR Grant \#FA9550-20-1-0054}.
This research was, in part, funded by the U.S. Government. 
The views and conclusions contained in this document are those of the 
authors and should not be interpreted as representing the official policies, 
either expressed or implied, of the U.S. Government.
\end{acks}



\bibliographystyle{ACM-Reference-Format} 
\bibliography{references}


\begin{thebibliography}{34}


\ifx \showCODEN    \undefined \def \showCODEN     #1{\unskip}     \fi
\ifx \showDOI      \undefined \def \showDOI       #1{#1}\fi
\ifx \showISBNx    \undefined \def \showISBNx     #1{\unskip}     \fi
\ifx \showISBNxiii \undefined \def \showISBNxiii  #1{\unskip}     \fi
\ifx \showISSN     \undefined \def \showISSN      #1{\unskip}     \fi
\ifx \showLCCN     \undefined \def \showLCCN      #1{\unskip}     \fi
\ifx \shownote     \undefined \def \shownote      #1{#1}          \fi
\ifx \showarticletitle \undefined \def \showarticletitle #1{#1}   \fi
\ifx \showURL      \undefined \def \showURL       {\relax}        \fi
\providecommand\bibfield[2]{#2}
\providecommand\bibinfo[2]{#2}
\providecommand\natexlab[1]{#1}
\providecommand\showeprint[2][]{arXiv:#2}

\bibitem[\protect\citeauthoryear{Adilov, Alexander, and Cunningham}{Adilov
  et~al\mbox{.}}{2015}]%
        {adilov2015economic}
\bibfield{author}{\bibinfo{person}{Nodir Adilov}, \bibinfo{person}{Peter~J
  Alexander}, {and} \bibinfo{person}{Brendan~M Cunningham}.}
  \bibinfo{year}{2015}\natexlab{}.
\newblock \showarticletitle{An economic analysis of earth orbit pollution}.
\newblock \bibinfo{journal}{\emph{Environmental and Resource Economics}}
  \bibinfo{volume}{60}, \bibinfo{number}{1} (\bibinfo{year}{2015}),
  \bibinfo{pages}{81--98}.
\newblock


\bibitem[\protect\citeauthoryear{Arslan, Marden, and Shamma}{Arslan
  et~al\mbox{.}}{2007}]%
        {arslan2007autonomous}
\bibfield{author}{\bibinfo{person}{Gürdal Arslan}, \bibinfo{person}{Jason~R.
  Marden}, {and} \bibinfo{person}{Jeff~S. Shamma}.}
  \bibinfo{year}{2007}\natexlab{}.
\newblock \showarticletitle{{Autonomous Vehicle-Target Assignment: A
  Game-Theoretical Formulation}}.
\newblock \bibinfo{journal}{\emph{Journal of Dynamic Systems, Measurement, and
  Control}} \bibinfo{volume}{129}, \bibinfo{number}{5} (\bibinfo{year}{2007}),
  \bibinfo{pages}{584--596}.
\newblock


\bibitem[\protect\citeauthoryear{Barman, Fawzi, and Ferm{\'e}}{Barman
  et~al\mbox{.}}{2020a}]%
        {barman2020tight2}
\bibfield{author}{\bibinfo{person}{Siddharth Barman}, \bibinfo{person}{Omar
  Fawzi}, {and} \bibinfo{person}{Paul Ferm{\'e}}.}
  \bibinfo{year}{2020}\natexlab{a}.
\newblock \showarticletitle{Tight approximation bounds for concave coverage
  problems}.
\newblock \bibinfo{journal}{\emph{arXiv preprint arXiv:2010.00970}}
  (\bibinfo{year}{2020}).
\newblock


\bibitem[\protect\citeauthoryear{Barman, Fawzi, Ghoshal, and
  G{\"u}rp{\i}nar}{Barman et~al\mbox{.}}{2020b}]%
        {barman2020tight}
\bibfield{author}{\bibinfo{person}{Siddharth Barman}, \bibinfo{person}{Omar
  Fawzi}, \bibinfo{person}{Suprovat Ghoshal}, {and} \bibinfo{person}{Emirhan
  G{\"u}rp{\i}nar}.} \bibinfo{year}{2020}\natexlab{b}.
\newblock \showarticletitle{Tight approximation bounds for maximum
  multi-coverage}. In \bibinfo{booktitle}{\emph{International Conference on
  Integer Programming and Combinatorial Optimization}}. Springer,
  \bibinfo{pages}{66--77}.
\newblock


\bibitem[\protect\citeauthoryear{Bil{\`o} and Vinci}{Bil{\`o} and
  Vinci}{2019}]%
        {bilo2019dynamic}
\bibfield{author}{\bibinfo{person}{Vittorio Bil{\`o}} {and}
  \bibinfo{person}{Cosimo Vinci}.} \bibinfo{year}{2019}\natexlab{}.
\newblock \showarticletitle{Dynamic taxes for polynomial congestion games}.
\newblock \bibinfo{journal}{\emph{ACM Transactions on Economics and Computation
  (TEAC)}} \bibinfo{volume}{7}, \bibinfo{number}{3} (\bibinfo{year}{2019}),
  \bibinfo{pages}{1--36}.
\newblock


\bibitem[\protect\citeauthoryear{Caragiannis}{Caragiannis}{2013}]%
        {caragiannis2013efficient}
\bibfield{author}{\bibinfo{person}{Ioannis Caragiannis}.}
  \bibinfo{year}{2013}\natexlab{}.
\newblock \showarticletitle{Efficient coordination mechanisms for unrelated
  machine scheduling}.
\newblock \bibinfo{journal}{\emph{Algorithmica}} \bibinfo{volume}{66},
  \bibinfo{number}{3} (\bibinfo{year}{2013}), \bibinfo{pages}{512--540}.
\newblock


\bibitem[\protect\citeauthoryear{Caragiannis, Kaklamanis, and
  Kanellopoulos}{Caragiannis et~al\mbox{.}}{2010}]%
        {caragiannis2010taxes}
\bibfield{author}{\bibinfo{person}{Ioannis Caragiannis},
  \bibinfo{person}{Christos Kaklamanis}, {and} \bibinfo{person}{Panagiotis
  Kanellopoulos}.} \bibinfo{year}{2010}\natexlab{}.
\newblock \showarticletitle{Taxes for linear atomic congestion games}.
\newblock \bibinfo{journal}{\emph{ACM Transactions on Algorithms (TALG)}}
  \bibinfo{volume}{7}, \bibinfo{number}{1} (\bibinfo{year}{2010}),
  \bibinfo{pages}{1--31}.
\newblock


\bibitem[\protect\citeauthoryear{Chandan, Paccagnan, and Marden}{Chandan
  et~al\mbox{.}}{2019}]%
        {chandan2019optimal}
\bibfield{author}{\bibinfo{person}{Rahul Chandan}, \bibinfo{person}{Dario
  Paccagnan}, {and} \bibinfo{person}{Jason~R Marden}.}
  \bibinfo{year}{2019}\natexlab{}.
\newblock \showarticletitle{When smoothness is not enough: toward exact
  quantification and optimization of the price-of-anarchy}. In
  \bibinfo{booktitle}{\emph{2019 IEEE 58th Conference on Decision and Control
  (CDC)}}. IEEE, \bibinfo{pages}{4041--4046}.
\newblock


\bibitem[\protect\citeauthoryear{Christodoulou, Koutsoupias, and
  Nanavati}{Christodoulou et~al\mbox{.}}{2004}]%
        {christodoulou2004coordination}
\bibfield{author}{\bibinfo{person}{George Christodoulou},
  \bibinfo{person}{Elias Koutsoupias}, {and} \bibinfo{person}{Akash Nanavati}.}
  \bibinfo{year}{2004}\natexlab{}.
\newblock \showarticletitle{Coordination mechanisms}. In
  \bibinfo{booktitle}{\emph{International Colloquium on Automata, Languages,
  and Programming}}. Springer, \bibinfo{pages}{345--357}.
\newblock


\bibitem[\protect\citeauthoryear{{\c{C}}olak, Lima, and
  Gonz{\'a}lez}{{\c{C}}olak et~al\mbox{.}}{2016}]%
        {ccolak2016understanding}
\bibfield{author}{\bibinfo{person}{Serdar {\c{C}}olak},
  \bibinfo{person}{Antonio Lima}, {and} \bibinfo{person}{Marta~C
  Gonz{\'a}lez}.} \bibinfo{year}{2016}\natexlab{}.
\newblock \showarticletitle{Understanding congested travel in urban areas}.
\newblock \bibinfo{journal}{\emph{Nature communications}} \bibinfo{volume}{7},
  \bibinfo{number}{1} (\bibinfo{year}{2016}), \bibinfo{pages}{1--8}.
\newblock


\bibitem[\protect\citeauthoryear{Conforti and Cornu{\'e}jols}{Conforti and
  Cornu{\'e}jols}{1984}]%
        {conforti1984submodular}
\bibfield{author}{\bibinfo{person}{Michele Conforti} {and}
  \bibinfo{person}{G{\'e}rard Cornu{\'e}jols}.}
  \bibinfo{year}{1984}\natexlab{}.
\newblock \showarticletitle{Submodular set functions, matroids and the greedy
  algorithm: tight worst-case bounds and some generalizations of the
  Rado-Edmonds theorem}.
\newblock \bibinfo{journal}{\emph{Discrete applied mathematics}}
  \bibinfo{volume}{7}, \bibinfo{number}{3} (\bibinfo{year}{1984}),
  \bibinfo{pages}{251--274}.
\newblock


\bibitem[\protect\citeauthoryear{Ding, Song, Morye, Farrell, and
  Roy-Chowdhury}{Ding et~al\mbox{.}}{2012}]%
        {ding2012collaborative}
\bibfield{author}{\bibinfo{person}{Chong Ding}, \bibinfo{person}{Bi Song},
  \bibinfo{person}{Akshay Morye}, \bibinfo{person}{Jay~A Farrell}, {and}
  \bibinfo{person}{Amit~K Roy-Chowdhury}.} \bibinfo{year}{2012}\natexlab{}.
\newblock \showarticletitle{Collaborative sensing in a distributed PTZ camera
  network}.
\newblock \bibinfo{journal}{\emph{IEEE Transactions on Image Processing}}
  \bibinfo{volume}{21}, \bibinfo{number}{7} (\bibinfo{year}{2012}),
  \bibinfo{pages}{3282--3295}.
\newblock


\bibitem[\protect\citeauthoryear{Feige}{Feige}{1998}]%
        {feige1998threshold}
\bibfield{author}{\bibinfo{person}{Uriel Feige}.}
  \bibinfo{year}{1998}\natexlab{}.
\newblock \showarticletitle{A threshold of ln n for approximating set cover}.
\newblock \bibinfo{journal}{\emph{Journal of the ACM (JACM)}}
  \bibinfo{volume}{45}, \bibinfo{number}{4} (\bibinfo{year}{1998}),
  \bibinfo{pages}{634--652}.
\newblock


\bibitem[\protect\citeauthoryear{Forbes}{Forbes}{2020}]%
        {forbes2020blackjack}
\bibfield{author}{\bibinfo{person}{Stephen Forbes}.}
  \bibinfo{year}{2020}\natexlab{}.
\newblock \bibinfo{title}{Blackjack}.
\newblock
\newblock


\bibitem[\protect\citeauthoryear{Gairing}{Gairing}{2009}]%
        {gairing2009covering}
\bibfield{author}{\bibinfo{person}{Martin Gairing}.}
  \bibinfo{year}{2009}\natexlab{}.
\newblock \showarticletitle{Covering games: Approximation through
  non-cooperation}. In \bibinfo{booktitle}{\emph{International Workshop on
  Internet and Network Economics}}. Springer, \bibinfo{pages}{184--195}.
\newblock


\bibitem[\protect\citeauthoryear{Giubilini}{Giubilini}{2019}]%
        {giubilini2019antibiotic}
\bibfield{author}{\bibinfo{person}{Alberto Giubilini}.}
  \bibinfo{year}{2019}\natexlab{}.
\newblock \showarticletitle{Antibiotic resistance as a tragedy of the commons:
  An ethical argument for a tax on antibiotic use in humans}.
\newblock \bibinfo{journal}{\emph{Bioethics}} \bibinfo{volume}{33},
  \bibinfo{number}{7} (\bibinfo{year}{2019}), \bibinfo{pages}{776--784}.
\newblock


\bibitem[\protect\citeauthoryear{Hardin}{Hardin}{1968}]%
        {hardin1968tragedy}
\bibfield{author}{\bibinfo{person}{Garrett Hardin}.}
  \bibinfo{year}{1968}\natexlab{}.
\newblock \showarticletitle{The Tragedy of the Commons}.
\newblock \bibinfo{journal}{\emph{Sci}} \bibinfo{volume}{162},
  \bibinfo{number}{3859} (\bibinfo{year}{1968}), \bibinfo{pages}{1243--1248}.
\newblock


\bibitem[\protect\citeauthoryear{Hart and Mas-Colell}{Hart and
  Mas-Colell}{2000}]%
        {hart2000simple}
\bibfield{author}{\bibinfo{person}{Sergiu Hart} {and} \bibinfo{person}{Andreu
  Mas-Colell}.} \bibinfo{year}{2000}\natexlab{}.
\newblock \showarticletitle{A simple adaptive procedure leading to correlated
  equilibrium}.
\newblock \bibinfo{journal}{\emph{Econometrica}} \bibinfo{volume}{68},
  \bibinfo{number}{5} (\bibinfo{year}{2000}), \bibinfo{pages}{1127--1150}.
\newblock


\bibitem[\protect\citeauthoryear{Hochbaum}{Hochbaum}{1996}]%
        {hochbaum1996approximating}
\bibfield{author}{\bibinfo{person}{Dorit~S Hochbaum}.}
  \bibinfo{year}{1996}\natexlab{}.
\newblock \showarticletitle{Approximating covering and packing problems: set
  cover, vertex cover, independent set, and related problems}.
\newblock In \bibinfo{booktitle}{\emph{Approximation algorithms for NP-hard
  problems}}. \bibinfo{pages}{94--143}.
\newblock


\bibitem[\protect\citeauthoryear{Immorlica, Li, Mirrokni, and Schulz}{Immorlica
  et~al\mbox{.}}{2009}]%
        {immorlica2009coordination}
\bibfield{author}{\bibinfo{person}{Nicole Immorlica}, \bibinfo{person}{Li~Erran
  Li}, \bibinfo{person}{Vahab~S Mirrokni}, {and} \bibinfo{person}{Andreas~S
  Schulz}.} \bibinfo{year}{2009}\natexlab{}.
\newblock \showarticletitle{Coordination mechanisms for selfish scheduling}.
\newblock \bibinfo{journal}{\emph{Theoretical computer science}}
  \bibinfo{volume}{410}, \bibinfo{number}{17} (\bibinfo{year}{2009}),
  \bibinfo{pages}{1589--1598}.
\newblock


\bibitem[\protect\citeauthoryear{Koutsoupias and Papadimitriou}{Koutsoupias and
  Papadimitriou}{2009}]%
        {koutsoupias2009worst}
\bibfield{author}{\bibinfo{person}{Elias Koutsoupias} {and}
  \bibinfo{person}{Christos Papadimitriou}.} \bibinfo{year}{2009}\natexlab{}.
\newblock \showarticletitle{Worst-case equilibria}.
\newblock \bibinfo{journal}{\emph{Computer science review}}
  \bibinfo{volume}{3}, \bibinfo{number}{2} (\bibinfo{year}{2009}),
  \bibinfo{pages}{65--69}.
\newblock


\bibitem[\protect\citeauthoryear{Lloyd}{Lloyd}{1833}]%
        {lloyd1833two}
\bibfield{author}{\bibinfo{person}{William~Forster Lloyd}.}
  \bibinfo{year}{1833}\natexlab{}.
\newblock \bibinfo{booktitle}{\emph{Two lectures on the checks to population}}.
\newblock


\bibitem[\protect\citeauthoryear{Marden and Roughgarden}{Marden and
  Roughgarden}{2014}]%
        {marden2014generalized}
\bibfield{author}{\bibinfo{person}{Jason~R Marden} {and} \bibinfo{person}{Tim
  Roughgarden}.} \bibinfo{year}{2014}\natexlab{}.
\newblock \showarticletitle{Generalized efficiency bounds in distributed
  resource allocation}.
\newblock \bibinfo{journal}{\emph{IEEE Trans. Automat. Control}}
  \bibinfo{volume}{59}, \bibinfo{number}{3} (\bibinfo{year}{2014}),
  \bibinfo{pages}{571--584}.
\newblock


\bibitem[\protect\citeauthoryear{Marden and Wierman}{Marden and
  Wierman}{2013}]%
        {marden2013distributed}
\bibfield{author}{\bibinfo{person}{Jason~R Marden} {and} \bibinfo{person}{Adam
  Wierman}.} \bibinfo{year}{2013}\natexlab{}.
\newblock \showarticletitle{Distributed welfare games}.
\newblock \bibinfo{journal}{\emph{Operations Research}} \bibinfo{volume}{61},
  \bibinfo{number}{1} (\bibinfo{year}{2013}), \bibinfo{pages}{155--168}.
\newblock


\bibitem[\protect\citeauthoryear{Murphey}{Murphey}{2000}]%
        {murphey2000target}
\bibfield{author}{\bibinfo{person}{Robert~A Murphey}.}
  \bibinfo{year}{2000}\natexlab{}.
\newblock \showarticletitle{Target-based weapon target assignment problems}.
\newblock In \bibinfo{booktitle}{\emph{Nonlinear assignment problems}}.
  \bibinfo{publisher}{Springer}, \bibinfo{pages}{39--53}.
\newblock


\bibitem[\protect\citeauthoryear{Paccagnan, Chandan, Ferguson, and
  Marden}{Paccagnan et~al\mbox{.}}{2020}]%
        {paccagnan2019incentivizing}
\bibfield{author}{\bibinfo{person}{Dario Paccagnan}, \bibinfo{person}{Rahul
  Chandan}, \bibinfo{person}{Bryce~L Ferguson}, {and} \bibinfo{person}{Jason~R
  Marden}.} \bibinfo{year}{2020}\natexlab{}.
\newblock \showarticletitle{Incentivizing efficient use of shared
  infrastructure: Optimal tolls in congestion games}.
\newblock \bibinfo{journal}{\emph{arXiv preprint arXiv:1911.09806v3}}
  (\bibinfo{year}{2020}).
\newblock


\bibitem[\protect\citeauthoryear{Paccagnan, Chandan, and Marden}{Paccagnan
  et~al\mbox{.}}{2019}]%
        {paccagnan2019utility1}
\bibfield{author}{\bibinfo{person}{Dario Paccagnan}, \bibinfo{person}{Rahul
  Chandan}, {and} \bibinfo{person}{Jason~R Marden}.}
  \bibinfo{year}{2019}\natexlab{}.
\newblock \showarticletitle{Utility Design for Distributed Resource
  Allocation-Part I: Characterizing and Optimizing the Exact Price of Anarchy}.
\newblock \bibinfo{journal}{\emph{IEEE Trans. Automat. Control}}
  (\bibinfo{year}{2019}).
\newblock


\bibitem[\protect\citeauthoryear{{Roughgarden}}{{Roughgarden}}{2014}]%
        {roughgarden2014barriers}
\bibfield{author}{\bibinfo{person}{T. {Roughgarden}}.}
  \bibinfo{year}{2014}\natexlab{}.
\newblock \showarticletitle{Barriers to Near-Optimal Equilibria}. In
  \bibinfo{booktitle}{\emph{2014 IEEE 55th Annual Symposium on Foundations of
  Computer Science}}. \bibinfo{pages}{71--80}.
\newblock


\bibitem[\protect\citeauthoryear{Roughgarden}{Roughgarden}{2015}]%
        {roughgarden2015intrinsic}
\bibfield{author}{\bibinfo{person}{Tim Roughgarden}.}
  \bibinfo{year}{2015}\natexlab{}.
\newblock \showarticletitle{Intrinsic robustness of the price of anarchy}.
\newblock \bibinfo{journal}{\emph{Journal of the ACM (JACM)}}
  \bibinfo{volume}{62}, \bibinfo{number}{5} (\bibinfo{year}{2015}),
  \bibinfo{pages}{1--42}.
\newblock


\bibitem[\protect\citeauthoryear{Saad, Han, Poor, and Basar}{Saad
  et~al\mbox{.}}{2012}]%
        {saad2012game}
\bibfield{author}{\bibinfo{person}{Walid Saad}, \bibinfo{person}{Zhu Han},
  \bibinfo{person}{H~Vincent Poor}, {and} \bibinfo{person}{Tamer Basar}.}
  \bibinfo{year}{2012}\natexlab{}.
\newblock \showarticletitle{Game-theoretic methods for the smart grid: An
  overview of microgrid systems, demand-side management, and smart grid
  communications}.
\newblock \bibinfo{journal}{\emph{IEEE Signal Processing Magazine}}
  \bibinfo{volume}{29}, \bibinfo{number}{5} (\bibinfo{year}{2012}),
  \bibinfo{pages}{86--105}.
\newblock


\bibitem[\protect\citeauthoryear{Shamma}{Shamma}{2008}]%
        {shamma2008cooperative}
\bibfield{author}{\bibinfo{person}{Jeff Shamma}.}
  \bibinfo{year}{2008}\natexlab{}.
\newblock \bibinfo{booktitle}{\emph{Cooperative control of distributed
  multi-agent systems}}.
\newblock \bibinfo{publisher}{John Wiley \& Sons}.
\newblock


\bibitem[\protect\citeauthoryear{Sviridenko, Vondr{\'a}k, and Ward}{Sviridenko
  et~al\mbox{.}}{2017}]%
        {sviridenko2017optimal}
\bibfield{author}{\bibinfo{person}{Maxim Sviridenko}, \bibinfo{person}{Jan
  Vondr{\'a}k}, {and} \bibinfo{person}{Justin Ward}.}
  \bibinfo{year}{2017}\natexlab{}.
\newblock \showarticletitle{Optimal approximation for submodular and
  supermodular optimization with bounded curvature}.
\newblock \bibinfo{journal}{\emph{Mathematics of Operations Research}}
  \bibinfo{volume}{42}, \bibinfo{number}{4} (\bibinfo{year}{2017}),
  \bibinfo{pages}{1197--1218}.
\newblock


\bibitem[\protect\citeauthoryear{Syrgkanis and Tardos}{Syrgkanis and
  Tardos}{2013}]%
        {syrgkanis2013composable}
\bibfield{author}{\bibinfo{person}{Vasilis Syrgkanis} {and}
  \bibinfo{person}{Eva Tardos}.} \bibinfo{year}{2013}\natexlab{}.
\newblock \showarticletitle{Composable and efficient mechanisms}. In
  \bibinfo{booktitle}{\emph{Proceedings of the forty-fifth annual ACM symposium
  on Theory of computing}}. \bibinfo{pages}{211--220}.
\newblock


\bibitem[\protect\citeauthoryear{Vetta}{Vetta}{2002}]%
        {vetta2002nash}
\bibfield{author}{\bibinfo{person}{Adrian Vetta}.}
  \bibinfo{year}{2002}\natexlab{}.
\newblock \showarticletitle{Nash equilibria in competitive societies, with
  applications to facility location, traffic routing and auctions}. In
  \bibinfo{booktitle}{\emph{The 43rd Annual IEEE Symposium on Foundations of
  Computer Science, 2002. Proceedings.}} IEEE, \bibinfo{pages}{416--425}.
\newblock


\end{thebibliography}


\newpage
\appendix

\section{Proof of Theorem~\ref{thm:informal}}
\label{app:informal}

The proof relies on a result in Chandan~et~al.~\cite{chandan2019optimal},
which we detail in the following proposition for the reader's convenience:
\begin{proposition}[Thm.~6~\cite{chandan2019optimal}]
\label{prop:linear_programming_based_approach}
Consider the set of resource allocation games with a maximum of $n$ players
induced by local welfare functions $W^1, \dots, W^m$.
Let $(F^j,\rho^j)$, $j=1,\dots,m$, be solutions to $m$ linear programs of the form\footnote{
    $\mc{I}(n)$ is defined as the set of all triplets $(x,y,z) \in \{0,1,\dots,n\}^3$ that satisfy:
    (i)~$1\leq x+y-z\leq n$ and $z\leq \min\{x,y\}$;
    and, (ii)~$x+y-z=n$ or $(x-z)(y-z)z=0$.
}
\begin{equation} \label{eq:original_linprog}
\begin{aligned}
    \min_{F, \rho} \quad & \rho \\
    \text{s.t.} \quad &
        W^j(y) - \rho W^j(x) + (x-z)F(x) - (y-z)F(x+1) \leq 0, \\
    & \hspace*{125pt} \forall (x,y,z) \in \mc{I}(n).
\end{aligned}
\end{equation}

Then, the local utility functions $F^1, \dots, F^m$ maximize the price of anarchy
and the corresponding price of anarchy is 
\( \poa(\mc{G}) = \min_{j\in\{1,\dots,m\}} \frac{1}{\rho^j} . \)
\end{proposition}


\subsection{Proof of Lemma~\ref{thm:coverage_solution}}
\label{app:coverage_solution}

\begin{proof}
We first dispense with the situation where $n \leq \beta$.
In this case, the local welfare function is identical to $W(x)=x$,
and thus the price of anarchy is 1 for choice of $F(x)=W(x)/x$.
For the remainder of the proof, we only consider $n > \beta$.

The remainder of the proof is structured as follows:
(i)~we introduce a relaxation of the linear program in Equation~(\ref{eq:original_linprog});
(ii)~in this relaxed linear program, 
we determine what are the most restrictive constraints for each $x\in\{1,\dots,n-1\}$;
(iii)~we show that a feasible solution to the relaxed linear program is nonincreasing, i.e., 
$F(x+1)\leq F(x)$, for every $x\in\{1,\dots,n-1\}$ such that $F(x)>1-\alpha$,
and $F(x+1)=1-\alpha$ otherwise;
(iv)~we show that $(F,\rho)$ as defined in the claim is a solution to the relaxed linear program 
for $n\to \infty$;
and (v)~we observe that $(F,\rho)$ as defined in the claim is feasible in the linear program
in Equation~(\ref{eq:original_linprog}) and thus a solution to this linear program as well.

\vspace*{.25cm}\noindent\emph{Relaxed linear program.}
First we consider a relaxation of the linear program in Equation~(\ref{eq:original_linprog}).
In this relaxed linear program, only the constraints where $z=\min\{0,x+y-n\}$ and $x,y \in \{0,\dots,n\}$ are retained.
Finally, we exclude the constraint with $y=0$, for all $x\in\{1,\dots,n-1\}$, resulting in the following
\emph{relaxed linear program}:

\begin{equation} \label{eq:relaxed_linprog}
\begin{aligned}
    &\max_{F \in \reals^n, \rho \in \reals} \> \rho \quad \text{subject to:} \\
    & W(y)-\rho W(x)+\min\{x,n-y\}F(x)-\min\{y,n-x\}F(x+1) \leq 0, \\
    & \hspace*{75pt} \forall(x,y) \in \{0,\dots,n\}\times\{1,\dots,n\} \cup (n,0).
\end{aligned}
\end{equation}

\vspace*{.25cm}\noindent\emph{Tightest constraints on $\rho$.}
We characterize what value $y \in \{1,\dots,n\}$ parameterizes the tightest constraint 
for each $x\in\{1,\dots,n-1\}$.
For any $x=1,\dots,n-1$, if $1 \geq F(x), F(x+1) \geq 1-\alpha$,
we observe that the constraint with $y=\beta$ is strictest.
For $y< \beta$, it holds that
\begin{align*}
    \rho W(x) \geq\> & \beta+\min\{x,n-\beta\}F(x)-\min\{\beta,n-x\}F(x+1) \\
              \geq\> & y+\min\{x,n-y\}F(x)-\min\{y,n-x\}F(x+1),
\end{align*}
where the final inequality holds 
when $x\leq n-\beta$ because $\beta-y \geq (\beta-y)F(x+1)$;
when $n-\beta< x \leq n-y$ because $\beta-y-(x+\beta-n)F(x) \geq n-x-y \geq (n-x-y)F(x+1)$ since $x+\beta-n>0$;
and when $n-y<x$ because $\beta-y\geq (\beta-y)F(x)$.
For constraints with $y> \beta$,
\begin{align*}
    & \rho W(x) \\ \geq\> & \alpha \beta+(1-\alpha)\beta+\min\{x,n-\beta\}F(x)-\min\{\beta,n-x\}F(x+1) \\
                   \geq\> & \alpha \beta+(1-\alpha)y+\min\{x,n-y\}F(x)-\min\{y,n-x\}F(x+1),
\end{align*}
where the final inequality holds 
when $x\leq n-y$ because $(y-\beta)F(x+1) \geq (y-\beta)(1-\alpha)$,
when $n-y<x\leq n-\beta$ because $(1-\alpha)(\beta-y)+(x+y-n)F(x)\geq(1-\alpha)(\beta+x-n)\geq(\beta+x-n)F(x+1)$ since $x+y-n>0\geq \beta+x-n$,
and when $n-\beta<x$ because $(1-\alpha)(\beta-y)\geq(\beta-y)F(x)$.

For any $x=1,\dots,n-1$, if $F(x) \geq 1-\alpha\geq F(x+1)$ and $n-x\geq \beta$,
then the constraint with $y=n-x$ is strictest among all constraints as,
for any $y \neq n-x$, it holds that
\begin{align*}
    & \rho W(x)\\  \geq\> & \alpha \beta+(1-\alpha)(n-x)+xF(x)-(n-x)F(x+1) \\
            \geq\> & \alpha \beta+(1-\alpha)y+\min\{x,n-y\}F(x)-\min\{y,n-x\}F(x+1) \\
            \geq\> & W(y)+\min\{x,n-y\}F(x)-\min\{y,n-x\}F(x+1),
\end{align*}
where the inequality holds because
$(1-\alpha)(n-x-y)\geq (n-x-y)F(x+1)$ when $x\leq n-y$
and $(1-\alpha)(n-x-y)\geq (n-x-y)F(x)$ when $x>n-y$.
For any $x=1,\dots,n-1$, if $F(x) \geq 1-\alpha\geq F(x+1)$ and $n-x < \beta$, 
then $y=\beta$ is strictest as for any $y \neq \beta$, it holds that
\begin{align*}
    \rho W(x) \geq\> & \beta+\min\{x,n-\beta\}F(x)-\min\{\beta,n-x\}F(x+1) \\
            \geq\> & W(y)+\min\{x,n-y\}F(x)-\min\{y,n-x\}F(x+1),
\end{align*}
where $\beta-y+(n-\beta-x)F(x)\geq (\beta-y)[1-F(x)] +(n-x-y)F(x) \geq (n-x-y)F(x+1)$ when $x\leq n-y$ since $y\leq n-x < \beta$,
$(\beta-y)[1-F(x)]\geq 0$ when $x > n-y$ and $n-x< y\leq \beta$,
and $(1-\alpha)(\beta-y)+(y-\beta)F(x)\geq 0$ when $y>\beta$ since $F(x)\geq 1-\alpha$.
For any $x=1,\dots,n-1$, if $F(x+1), F(x)\leq 1-\alpha$,
then the constraint with $y=n$ is strictest among all constraints as,
for any $y<n$, it holds that
\begin{align*}
    & \rho W(x) \\ \geq\> & \alpha \beta+(1-\alpha)n+\min\{x,0\}F(x)-\min\{n,n-x\}F(x+1) \\
            \geq\> & \alpha \beta+(1-\alpha)y+\min\{x,n-y\}F(x)-\min\{y,n-x\}F(x+1) \\
            \geq\> & W(y)+\min\{x,n-y\}F(x)-\min\{y,n-x\}F(x+1),
\end{align*}
where the second last inequality holds because 
$(n-y)[1-\alpha-F(x+1)] \geq x[F(x)-F(x+1)]$ when $x\leq n-y$
and $(n-y)(1-\alpha)\geq (n-y)F(x)$ when $x>n-y$.

Thus, if $F(x)\geq F(x+1)\geq 1-\alpha$, it is sufficient to consider only the constraint 
with $y=\beta$ and $z=\max\{0,x+\beta-n\}$.
If $F(x)\geq 1-\alpha\geq F(x+1)$ and $n-x > \beta$, it is sufficient to consider only the constraint
with $y=n-x$ and $z=0$.
Otherwise, if $F(x)\geq 1-\alpha\geq F(x+1)$ and $n-x< \beta$, 
then we consider only the constraint $y=\beta$ and $z=\max\{0,x+y-n\}$.
Finally, if $1-\alpha \geq F(x) \geq F(x+1)$, then the constraint with $y=n$ and $z=x$ is the strictest.

\vspace*{.25cm}\noindent\emph{Proof that a solution to Equation~(\ref{eq:relaxed_linprog}) has $F$ `nonincreasing'.}
For this portion of the proof, consider a function $F$ defined for
any given $\rho>1$ as follows:
$F(1) = W(1)$ and, for all $x\in\{1,\dots,n-1\}$,
\[ F(x+1) = \max_{y\in\{1,\dots,n\}} \frac{\min\{x,n-y\}F(x) - W(x) \rho + W(y)}{\min\{y,n-x\}}. \]
For conciseness, we will use the shorthand
$\kappa_x=\frac{\min\{x,n-y^*\}}{\min\{y^*,n-x\}}$,
$\lambda_x= \frac{W(x)}{\min\{y^*,n-x\}}$ and
$\mu_x=\frac{W(y^*)}{\min\{y^*,n-x\}}$ where $y^*\in\{1,\dots,n\}$ maximizes the above expression
for each $x\in\{1,\dots,n\}$.
Thus, $F(x+1)=\kappa_x F(x)-\lambda_x \rho+\mu_x$ for each $x\in\{1,\dots,n\}$.

We assume, by contradiction, that we are given $\rho$ such that $F$ is increasing at some index, 
i.e., $F(x+1)>F(x)$ for $x\in\{1,\dots,n\}$. 
Lemma~\ref{lem:f_continues_increasing} in Appendix~\ref{app:f_continues_increasing} shows that, if this is the case, 
then $F$ must continue to increase, so that $F(n) > F(n-1)$.
We wish to show the following:
(i)~if $F$ first increases at a point $x\in\{1,\dots,n-1\}$ where $F(x)>1-\alpha$, 
this leads to a contradiction for the value of $\rho$;
and, (ii)~if $F$ first increases at a point $x$ where $F(x)\leq 1-\alpha$, 
then either $F(j)\leq 1-\alpha$ for all $j\geq x$ is feasible or $(F,\rho)$ is infeasible.
It is important to note that the value $nF(n)/W(n)$ must be bounded, otherwise
\[ \rho \geq \max_{y \in \{0,\dots,n\}} \frac{W(y)+(n-y)F(n)}{W(n)} \geq \frac{nF(n)}{W(n)} \to \infty. \]
This is a contradiction as the price of anarchy will be at least $0.5$,
even if we use the marginal contribution utility~\cite{vetta2002nash}.
Since we are optimizing for the price of anarchy (i.e., $1/\rho$), 
we need only consider values of $\rho$ no greater than $2$.

Observe that if $F$ first increases at some point $x\in\{1,\dots,n\}$ such that $F(x)>1-\alpha$, 
then $F(n) > F(n-1) > 1-\alpha$ and
\begin{align*}
    \rho &= \max_{y \in \{0,\dots,n\}} \frac{W(y)+(n-y)F(n)}{W(n)} \\
         &\geq \max_{y \in \{1,\dots,n\}} \frac{W(y)+(n-y)F(n)}{W(n)} \\
         &= \max_{y \in \{1,\dots,n-1\}} \frac{W(y)+(n-y)F(n)}{W(n)} \\
         &> \max_{y \in \{1,\dots,n-1\}} \frac{W(y)+(n-y)F(n-1)}{W(n)},
\end{align*}
where the first inequality holds because we reduce the domain of maximization,
the second equality holds because $W(n-1)+F(n) > W(n)$ since $n>\beta$ and $F(n)>1-\alpha$,
and the final inequality holds because $F(n) > F(n-1)$.
Since $y\leq \beta \leq n-1$ corresponds with the strictest constraints when $F(n), F(n-1) \geq 1-\alpha$,
we can substitute 
\[ \max_{y \in \{1,\dots,n-1\}} [W(y)+(n-y)F(n-1)] = F(n)+\rho W(n-1) \] 
in the former bound on $\rho$ to get $\rho > [F(n)+\rho W(n-1)] / W(n)$.
Since $n>\beta$, this implies that
\begin{align*}
    \begin{cases}
        F(n) < 0 \quad &\text{if } \alpha=1, \\
        \rho > \frac{F(n)}{W(n)-W(n-1)} = \frac{F(n)}{1-\alpha}\quad &\text{if } \alpha\in[0,1)
    \end{cases}
\end{align*}
For $\alpha=1$ this is a contradiction, since we have that $F(n)>1-\alpha=0$.
For the remaining $\alpha\in[0,1)$, we want to prove that this 
also gives rise to a contradiction.
To do so, we show that 
\[ \frac{F(n)}{1-\alpha} \geq \max_{y \in \{0,\dots,n\}} \frac{W(y)+(n-y)F(n)}{W(n)} = \rho. \]
Observe that $y=\beta$ maximizes the right-hand side if $1-\alpha<F(n)\leq 1$,
and $y=0$ maximizes the right-hand side if $F(n) > 1$.
For $F(n)\in (1-\alpha,1]$ and $y=\beta$, it holds that
\begin{align*}
        & \frac{F(n)}{1-\alpha}-\frac{W(\beta)+(n-\beta)F(n)}{W(n)} \geq 0\\
    \iff& F(n) W(n) - (1-\alpha)[W(\beta)+(n-\beta)F(n)] \geq 0\\
    \impliedby& F(n)\frac{W(n)-W(\beta)}{n-\beta} = (1-\alpha)F(n)\\
\end{align*}
where the first and second line are equivalent because $\alpha\in[0,1)$
and the third line implies the second because $F(n)>1-\alpha$.
The final inequality holds because $[W(n)-W(\beta)]=(n-\beta)(1-\alpha)$ and $n>\beta$.
For $F(n) > 1$ and $y=0$, it holds that
\begin{align*}
    & \frac{F(n)}{1-\alpha}-\frac{nF(n)}{W(n)} \geq 0
\end{align*}
since $W(n)/n \geq 1-\alpha$ by definition.
Thus, in the above reasoning, we have shown that, if $F$ first increases at 
a point $x\in\{1,\dots,n\}$ and $F(x)>1-\alpha$, then, 
if $\alpha=1$, it holds that $1-\alpha<F(x)<\dots<F(n)<1-\alpha$;
and, if $\alpha\in[0,1)$, it holds that
\[ \rho > \frac{F(n)}{1-\alpha} \geq \max_{y \in \{0,\dots,n\}} \frac{W(y)+(n-y)F(n)}{W(n)} = \rho, \]
which is a contradiction.

Now we consider the scenario where we are at a point $x$ such that $F(x)\leq 1-\alpha$ 
and $F$ is monotonically nonincreasing before $x$.
We show that 
either selecting $F(x)=\dots=F(n)=1-\alpha$ is feasible for $\rho$ 
or that the value $\rho$ is infeasible.
We first consider the case where the strictest constraint on the value of $F(x+1)$ has $y\leq n-x$
and show that $F(x+1)$ cannot be greater than $F(x)\leq 1-\alpha$.
In the proof of Lemma~\ref{lem:f_continues_increasing} in Appendix~\ref{app:f_continues_increasing}, 
we showed that if $y\leq n-x$, $F(x+1)>F(x)$ and $F(x)\leq F(x-1)$,
then it must hold that $F(x) > [W(x)-W(x-1)] \rho \geq 1-\alpha$.
As we have assumed $F(x) \leq 1-\alpha$, it must be that $F(x+1) \leq F(x)\leq 1-\alpha$ if $y\leq n-x$.
We complete our reasoning for the case when $y>n-x$ corresponds to the strictest constraint 
on the value of $F(x+1)$.
We showed above that if $F(x) \leq 1-\alpha$ and $x>n-y$,
then the strictest constraint is parameterized by $y=n$.
For any $x\geq \beta$, it must hold that $F(x+1)\geq-W(x) \rho/(n-x) + W(n)/(n-x)$.
Since $-W(x) \rho/(n-x) + W(n)/(n-x) \leq 1-\alpha$ the constraint is satisfied for choice of $F(x+1) \leq 1-\alpha$.
Else, if $x<\beta$, since $\beta<n$, $F(x+1)>1-\alpha$ implies that $F(n)>F(n-1)>1-\alpha$, since $n-1 \geq \beta$.
We already proved above that this scenario leads to a contradiction on the value of $\rho$.
Repeating this reasoning for all $j>x$ such that $F(x)\leq 1-\alpha$, we argue that $F(j)= 1-\alpha$ is feasible.
Since the strictest constraint for each $F(j), F(j+1) \leq 1-\alpha$ has $y=n$, 
there is no recursion and the optimal value $\rho$ has no dependence on the values of $F(x), \dots, F(n)$,
even if it begins increasing.
We have also shown that the lower bound on $F$ is lower than or equal to $1-\alpha$ for any feasible $\rho$,
and so, $F$ with $F(j)=1-\alpha$, for all $j\in\{x,\dots,n\}$, must be feasible.

For any feasible $\rho$, we have successfully shown that $F(x)$ must be nonincreasing when it is greater than $1-\alpha$,
and that $F(x)=\dots=F(n)=1-\alpha$ is feasible otherwise.
This concludes this part of the proof.

\vspace*{.25cm}\noindent\emph{Proof that $(F,\rho)$ solves Equation~(\ref{eq:relaxed_linprog}).}
We begin by showing that $(F,\rho)$ as defined in the claim 
are feasible.
For $x=0$, the constraints in Equation~(\ref{eq:relaxed_linprog}) read as
$F(1) \geq W(y) / \min\{y,n-x\}$, for all $y=1,\dots,n$, which is satisfied for $F(1)=W(1)$.
Now consider $(x,y) \in\{1,\dots,n-1\} \times \{1,\dots,n\}$.
In the above reasoning, we showed that a feasible $(F,\rho)$ within Equation~(\ref{eq:relaxed_linprog}) 
will have $F$ nonincreasing while $F(x)>1-\alpha$ and $F(x)=1-\alpha$ otherwise.
Furthermore, we showed that when $F(x)\geq F(x+1)> 1-\alpha$ or when 
$F(x)\geq 1-\alpha\geq F(x+1)$ and $n-x< \beta$, then the strictest constraint has $y=\beta$.
Observe that
$\kappa_x = \min\{x,n-\beta\}/\min\{\beta,n-x\}$, $\lambda_x = W(x)/\min\{\beta,n-x\}$ 
and $\mu_x = \beta/\min\{\beta,n-x\}$ correspond with the recursive definition 
of $F(x+1)$.

We showed above that $F(x+1)=1-\alpha$ when $\kappa_x F(x) - \lambda_x \rho + \mu_x \leq 1-\alpha$,
is feasible as long as $\rho$ is feasible, since the values of $F$ less than or equal to $1-c$
have no impact on the optimal value of the relaxed linear program.
Consider the expression for $\rho$ that can be obtained
by completing the recursion as follows, 
\[ 1-\alpha = F(\hat x+1) \geq \Pi^{\hat x}_{u=1} \kappa_u F(1) 
            + \sum^{\hat x-1}_{u=1} (\Pi^{\hat x}_{v=u+1} \kappa_v) (\mu_u-\lambda_u \rho)
            + \mu_{\hat x} - \lambda_{\hat x} \rho. \]
Rearranging this expression, we obtain,
\begin{align*}
    \rho &\geq \frac{\Pi^{\hat x}_{u=1} \kappa_u F(1) + \sum^{\hat x-1}_{u=1} (\Pi^{\hat x}_{v=u+1} \kappa_v)\mu_u+ \mu_{\hat x} + \alpha -1}
                    {\sum^{\hat x-1}_{u=1} (\Pi^{\hat x}_{v=u+1} \kappa_v) \lambda_u + \lambda_{\hat x}}.
\end{align*}

Observe that for $n \to \infty$, $\min\{x,n-\beta\} = x$ and $\min\{\beta,n-x\}=\beta$.
Thus, the above expression for $\rho$ simplifies to
\begin{align*}
    \rho &\geq \frac{\frac{\hat x!}{\beta^{\hat x}}+\sum^{\hat x-1}_{u=1} \frac{\hat x!}{j!}\frac{1}{\beta^{\hat x-j}}+1+\alpha-1}
                    {\sum^\beta_{j=1} \frac{\hat x!}{j!\beta^{\hat x-j}}\frac{j}{\beta}+\sum^{\hat x-1}_{j=\beta+1}\frac{\hat x!}{j!\beta^{\hat x-j}}\frac{\alpha \beta+(1-\alpha)j}{\beta}} \\
            &= \frac{1+\sum^{\hat x-1}_{j=1}\frac{\beta^j}{j!}+\alpha \frac{\beta^{\hat x}}{\hat x!}}
                    {\sum^{\beta-1}_{j=0}\frac{\beta^j}{j!}+\sum^{\hat x-1}_{j=\beta+1}\frac{\beta^j}{j!}\frac{\alpha \beta+(1-\alpha)j}{\beta} } 
                    = \frac{e^\beta}{e^\beta-\alpha \frac{\beta^\beta}{\beta!}}.
\end{align*}
Noting that $\poa = 1/\rho$ concludes this part of the proof.

\vspace*{.25cm}\noindent\emph{Feasibility of $(F,\rho)$ in Equation~(\ref{eq:original_linprog}).}
To conclude the proof, we simply observe that since $F(x)$ is nonincreasing for all $x$,
the strictest constraints in the linear program in Equation~(\ref{eq:original_linprog})
correspond with the choice of $z=\min\{0,x+y-n\}$.
Thus, since $(F,\rho)$ is a solution to the relaxed linear program 
and feasible in the original linear program,
it must also be a solution to the original.
\end{proof}


\subsection{Proof of Lemma~\ref{lem:f_continues_increasing}}
\label{app:f_continues_increasing}

\begin{lemma}
\label{lem:f_continues_increasing}
Let $W$ be a nonnegative, nondecreasing concave function, 
and let $\rho \geq 1$ be a given parameter.
Further, define the function $F$ such that $F(1)=W(1)$ and
\begin{equation} \label{eq:fj_definition}
    F(j+1) := \max_{\ell \in \{1,\dots,n\}} 
        \frac{\min\{j,n-\ell\}F(j)-W(j)\rho+W(\ell)}
             {\min\{\ell,n-j\}},
\end{equation}
for all $j=1,\dots,n-1$.
Then, for the lowest value $\hat j=1,\dots,n-1$ such that $F(\hat j+1)>F(\hat j)$,
it must hold that $F(j+1)>F(j)$ for all $j=\hat j,\dots,n-1$.
\end{lemma}

\begin{proof}
The proof is presented in two parts as follows:
in part~(i), we identify an inequality that must hold given that $F(\hat j+1)>F(\hat j)$ for $1 \leq \hat j\leq n-1$
as defined in the claim;
and, in part~(ii), we use a recursive argument to prove that $F(j+1)>F(j)$ holds for all $\hat j+1\leq j\leq n-1$,
using the inequality we derived in part~(i).
\vspace*{.25cm}\noindent\emph{Part (i).}
We define $\ell^*_j$ as one of the arguments that minimizes the right-hand side of
Equation~(\ref{eq:fj_definition}) for each $x=1,\dots,n-1$.
By assumption, it must hold that $F(\hat j+1)>F(\hat j)$, which implies that
\begin{align*}
    F(\hat j) < \max_{1 \leq \ell \leq n} &
            \frac{\min\{\hat j,n-\ell\}}{\min\{\ell,n-\hat j\}} F(\hat j)
            - \frac{W(\hat j)}{\min\{\ell,n-\hat j\}} \rho
            + \frac{W(\ell)}{\min\{\ell,n-\hat j\}} \\
        = \max_{1 \leq \ell \leq n} &
            \frac{\min\{\hat j-1,n-\ell\}}{\min\{\ell,n-\hat j+1\}} F(\hat j-1)
            - \frac{W(\hat j-1)}{\min\{\ell,n-\hat j+1\}} \rho \\
        &   + \frac{W(\ell)}{\min\{\ell,n-\hat j+1\}} 
            + \frac{\min\{\hat j,n-\ell\}}{\min\{\ell,n-\hat j\}} F(\hat j) \\
        &   - \frac{\min\{\hat j-1,n-\ell\}}{\min\{\ell,n-\hat j+1\}} F(\hat j-1)
            - \frac{W(\hat j)}{\min\{\ell,n-\hat j\}} \rho \\
        &   + \frac{W(\hat j-1)}{\min\{\ell,n-\hat j+1\}} \rho
            + \frac{W(\ell)}{\min\{\ell,n-\hat j\}} \\
        &   - \frac{W(\ell)}{\min\{\ell,n-\hat j+1\}},
\end{align*}
where the strict inequality holds by definition of $F(\hat j+1)$.
Recall that
\begin{equation*}
\begin{aligned}
    F(\hat j) := \max_{1 \leq \ell \leq n}
&\frac{\min\{\hat j-1,n-\ell\}}{\min\{\ell,n-\hat j+1\}} F(\hat j-1)
- \frac{W(\hat j-1)}{\min\{\ell,n-\hat j+1\}} \rho \\
&+ \frac{W(\ell)}{\min\{\ell,n-\hat j+1\}}.
\end{aligned}
\end{equation*}
%
Thus, if $\ell^*_{\hat j} \leq n-\hat j$, 
the above strict inequality with $F(\hat j)$ can only be satisfied if
\[ F(\hat j+1) > F(\hat j) \geq \hat j F(\hat j) - (\hat j-1) F(\hat j-1) 
> [ W(\hat j) - W(\hat j-1) ] \cdot \rho. \]
Similarly, if $\ell^*_{\hat j} \geq n-\hat j+1$, then it must hold that
\begin{equation*}
\begin{aligned}
& (n-\ell^*_{\hat j}) \bigg[ \frac{F(\hat j)}{n-\hat j} - \frac{F(\hat j-1)}{n-\hat j+1} \bigg]
+ \bigg[ \frac{1}{n-\hat j} - \frac{1}{n-\hat j+1} \bigg] W(\ell_{\hat j}) \\
& > \bigg[ \frac{W(\hat j)}{n-\hat j} 
- \frac{W(\hat j-1)}{n-\hat j+1} \bigg] \cdot \rho \\
\implies & \bigg[ \frac{1}{n-\hat j} - \frac{1}{n-\hat j+1} \bigg] 
[(n-\ell^*_{\hat j}) F(\hat j) + W(\ell_{\hat j})] \\
& > \bigg[ \frac{W(\hat j)}{n-\hat j} 
- \frac{W(\hat j-1)}{n-\hat j+1} \bigg] \cdot \rho \\
\iff & F(\hat j+1)>[W(\hat j) - W(\hat j-1)] \rho,
\end{aligned}
\end{equation*}
where the first line implies the second line because $F(\hat j) \leq F(\hat j-1)$, 
by the definition of $\hat j$ in the claim,
and the second line is equivalent to the third by the definitions of 
$F(\hat j+1)$ and $\ell^*_{\hat j}$.
This concludes part~(i) of the proof.
\vspace{.2cm} \noindent {\bf Part~(ii).} 
In this part of the proof, we show by recursion that if $F(\hat j+1)>F(\hat j)$,
then $F(j+1)>F(j)$ for all $j=\hat j+1,...,n-1$.
We do so by showing that, if $F(j)>F(j-1)>\dots>F(\hat j+1)$ for any $\hat j+1 \leq j \leq n-1$,
then it must hold that $F(j+1)>F(j)$.
Thus, in the following reasoning, we assume that $\hat j+1 \leq j \leq n-1$,
and that $F(j)>F(j-1)>\dots>F(\hat j+1)$.
We begin with the scenario in which $\ell^*_{j-1} < n-j+1$,
which gives us that $\ell^*_{j-1} \leq n-j$.
Recall that
\begin{align*}
F(j+1) := \max_{1 \leq \ell_j \leq n}&
\frac{\min\{j,n-\ell_j\}}{\min\{\ell_j,n-j\}} F(j+1) 
- \frac{W(j)}{\min\{\ell_j,n-j\}} \rho \\
& + \frac{W(\ell_j)}{\min\{\ell_j,n-j\}}.
\end{align*}
Thus, it must hold that
\begin{align*}
& \quad F(j+1) \\
& = \max_{1 \leq \ell_j \leq n}
    \frac{\min\{j-1,n-\ell_j\}}{\min\{\ell_j,n-j+1\}} F(j-1)
    - \frac{W(j-1)}{\min\{\ell_j,n-j+1\}} \rho \\
& \hspace*{40pt}  + \frac{W(\ell_j)}{\min\{\ell_j,n-j+1\}} 
    + \frac{\min\{j,n-\ell_j\}}{\min\{\ell_j,n-j\}} F(j+1) \\
& \hspace*{40pt}  - \frac{\min\{j-1,n-\ell_j\}}{\min\{\ell_j,n-j+1\}} F(j-1)
    - \frac{W(j)}{\min\{\ell_j,n-j\}} \rho \\
& \hspace*{40pt}  + \frac{W(j-1)}{\min\{\ell_j,n-j+1\}} \rho
    + \frac{W(\ell_j)}{\min\{\ell_j,n-j\}} \\
& \hspace*{40pt}  - \frac{W(\ell_j)}{\min\{\ell_j,n-j+1\}}\\
&\geq F(j) + \frac{j}{\ell^*_{j-1}} F(j) - \frac{j-1}{\ell^*_{j-1}} F(j-1)
- \frac{W(j)-W(j-1)}{\ell^*_{j-1}} \rho \\
&> F(j) + \frac{1}{\ell^*_{j-1}} F(\hat j+1) - \frac{1}{\ell^*_{j-1}}[W(j)-W(j-1)] \rho \\
&> F(j),
\end{align*}
where the first inequality holds by evaluating the maximization at $\ell_j = \ell^*_{j-1}$,
the second inequality holds because $F(j)>F(j-1)$ and $F(j) \geq F(\hat j+1)$, by assumption, 
and the final inequality holds by the identity we showed in part~(i)
and because $W(\cdot)$ is concave.
Next, consider the scenario in which $\ell^*_{j-1} > n-j+1$.
Observe that
\begin{align*}
& \quad F(j+1) \\ &\geq F(j) 
+ (n-\ell^*_{j-1}) \bigg[ \frac{F(j)}{n-j} - \frac{F(j-1)}{n-j+1} \bigg] 
 \\
&\qquad + \bigg[ \frac{1}{n-j} - \frac{1}{n-j+1} \bigg] W(\ell^*_{j-1}) 
        - \frac{W(j)}{n-j} \rho + \frac{W(j-1)}{n-j+1} \rho \\
&> F(j)
+ \bigg[ \frac{1}{n-j} - \frac{1}{n-j+1} \bigg] 
[(n-\ell^*_{j-1}) F(j-1) + W(\ell^*_{j-1})] \\
&\qquad - \frac{W(j)}{n-j} \rho + \frac{W(j-1)}{n-j+1} \rho \\
&= F(j) 
+ \bigg[ \frac{1}{n-j}-\frac{1}{n-j+1} \bigg] [ (n-j+1)F(j) + W(j-1)\rho ]\\
&\qquad - \frac{W(j)}{n-j} \rho + \frac{W(j-1)}{n-j+1} \rho \\
&\geq F(j) + \frac{1}{n-j}F(\hat j+1) - \frac{1}{n-j} [W(j)-W(j-1)] \rho \\
&> F(j),
\end{align*}
where the first inequality holds by evaluating the maximization at $\ell_j = \ell^*_{j-1}$,
the second inequality holds because $F(j)>F(j-1)$, by assumption, 
the equality holds by the definitions of $F(j)$ and $\ell^*_{j-1}$,
the third inequality holds because $F(j)\geq F(\hat j+1)$, by assumption,
and the final inequality holds by the identity we showed in part~(i)
and because $W(\cdot)$ is concave.
Finally, we consider the scenario in which $\ell^*_{j-1} = n-j+1$.
Observe that
\begin{align*}
& \quad F(j+1) \\
&\geq F(j) + \frac{j-1}{n-j} F(j) - \frac{j-1}{n-j+1} F(j-1)
- \frac{W(j)}{n-j} \rho + \frac{W(j-1)}{n-j+1} \rho\\
&\qquad + \bigg[ \frac{1}{n-j} - \frac{1}{n-j+1} \bigg] W(n-j+1) \\
&> F(j)
+ \bigg[ \frac{1}{n-j} - \frac{1}{n-j+1} \bigg] 
[(n-\ell^*_{j-1}) F(j-1) + W(\ell^*_{j-1})] \\
&\qquad - \frac{W(j)}{n-j} \rho + \frac{W(j-1)}{n-j+1} \rho \\
&> F(j),
\end{align*}
where the first inequality holds by evaluating the maximization at $\ell_j = \ell^*_{j-1}$,
the second inequality holds because $F(j)>F(j-1)$, by assumption,
and the final inequality holds by the same reasoning as for $\ell^*_{j-1} > n-j+1$.
\end{proof}


\subsection{Proof of Lemma~\ref{lem:nonnegative_lincomb}}
\label{app:nonnegative_lincomb}

\begin{proof}
The proof is by construction.
Define coefficients $\eta_1 := [ 2 W(1) - W(2) ]/c$,
$\eta_j := [ 2 W(j) - W(j-1) - W(j+1) ]/c$, $j=2,\dots,n-1$,
and $\eta_n := W(1) - \sum_{j=1}^{n-1} \eta_j = W(1) - [ W(1) + W(n-1) - W(n) ]/c$.
It is straightforward to verify that $\eta_j \geq 0$ for all $k=1,\dots,n$
recalling that $W(0)=0$ and $W(x)$ is nonnegative, nondecreasing concave for $x\geq 0$.
We defer the proof that $W(x)=\sum_{k=1}^n \eta_k \cdot V^c_k(x)$ for all $x = 1,\dots,n$
to the proof of Corollary~\ref{cor:generalized_candidate},
where one need only substitute $W^{ub}(x)=x$ and $W^{lb}(x)=V^c_1(x)$, for $x\geq 0$.
\end{proof}


\section{Proof of Corollary~\ref{cor:generalized_candidate}.}
\label{app:generalized_candidate}

\begin{proof}
First, observe that there must exist functions $W^{ub}$ and $W^{lb}$.
Simply observe that $W^{ub}(x)=x$ and $W^{lb}(x)=V^1_1(x)=\min\{x,1\}$
are valid for any set of nonnegative, nondecreasing concave functions.

The rest of the proof follows by construction.
Define the coefficients $\eta_j$, $j=0,\dots,n$, as follows:
\[ \eta_1 = \frac{W^{ub}(2)-W^{ub}(1)-W(2)+W(1)}{W^{ub}(2)-W^{ub}(1)-W^{lb}(2)+W^{lb}(1)},\]
\[ \eta_j = \frac{W^{ub}(j+1)-W^{ub}(j)-W(j+1)+W(j)}{W^{ub}(j+1)-W^{ub}(j)-W^{lb}(j+1)+W^{lb}(j)} - \sum^{j-1}_{k=1} \eta_k, \]
for $j=2,\dots,n-1$ and $\eta_n = 1-\sum^{n-1}_{k=1} \eta_k$.

First, we prove that the coefficients $\eta_1, \dots, \eta_n$ are nonnegative.
It is simple to see that $\eta_1 \geq 0$ since $W^{ub}(2)-W^{ub}(1) \geq W(2)-W(1) \geq W^{lb}(2)-W^{lb}(1)$.
Similarly, $\eta_n \geq 0$ since $\eta_n = 1 - [ W^{ub}(n)-W^{ub}(n-1)-W(n)+W(n-1) ] / [ W^{ub}(n)-W^{ub}(n-1)-W^{lb}(n)+W^{lb}(n-1) ]$.
Finally, for any $j \in \{2, \dots, n-1\}$,
\begin{align*}
    \eta_j &= \frac{W^{ub}(j+1)-W^{ub}(j)-W(j+1)+W(j)}{ W^{ub}(j+1)-W^{ub}(j)-W^{lb}(j+1)+W^{lb}(j)} \\
            & \qquad - \frac{W^{ub}(j)-W^{ub}(j-1)-W(j)+W(j-1)}{W^{ub}(j)-W^{ub}(j-1)-W^{lb}(j)+W^{lb}(j-1)} \\
            &\geq \frac{W^{ub}(j+1)-2W^{ub}(j)+W^{ub}(j-1)}{W^{ub}(j)-W^{ub}(j-1)-W^{lb}(j)+W^{lb}(j-1)} \\
            & \qquad - \frac{W(j+1)-2W(j)W(j)+W(j-1)}{W^{ub}(j)-W^{ub}(j-1)-W^{lb}(j)+W^{lb}(j-1)} \geq 0,
\end{align*}
where the equality holds by definition, 
the first inequality holds because $W^{lb}(j+1)-2W^{lb}(j)+W^{lb}(j-1) \geq W^{ub}(j+1)-2W^{ub}(j)+W^{ub}(j-1)$
and the final inequality holds because $W^{ub}(j+1)-2W^{ub}(j)+W^{ub}(j-1) \geq W(j+1)-2W(j)+W(j-1)$.

We conclude the proof by observing that, for all $x=1,\dots,n$,
\begin{align*}
    &\quad \sum^{x-1}_{k=1} \eta_k W^{lb}(x) + \sum^{x-1}_{k=1} \eta_k [W^{ub}(k)-W^{lb}(k)] + \sum^n_{k=x} \eta_k W^{ub}(x) \\
    &= \frac{W^{ub}(x)-W^{ub}(x-1)-W(x)+W(x-1)}{W^{ub}(x)-W^{ub}(x-1)-W^{lb}(x)+W^{lb}(x-1)} W^{lb}(x) \\
    &\quad + \sum^{x-1}_{k=1}\eta_k [W^{ub}(k)-W^{lb}(k)] \\
    &\quad + \left[1-\frac{W^{ub}(x)-W^{ub}(x{-}1)-W(x)+W(x{-}1)}{W^{ub}(x)-W^{ub}(x{-}1)-W^{lb}(x)+W^{lb}(x{-}1)}\right] W^{ub}(x) \\
    &= W^{ub}(x) - W^{ub}(x)-W^{ub}(x-1)-W(x)+W(x-1) \\
    &\quad + \sum^{x-2}_{k=1} \eta_k [W^{ub}(k)-W^{lb}(k)-W^{ub}(x-1)+W^{lb}(x-1)] \\
    &= W(x) + \sum^{x-2}_{k=1} \eta_k [W^{ub}(k)-W^{lb}(k)-W^{ub}(x-1)+W^{lb}(x-1)] \\
    &\quad + W^{ub}(x-1) - W(x-1) \\
    &= W(x),
\end{align*}
where the final equality holds once the expression is simplified for the remaining $\eta_k$ values.
\end{proof}

\end{document}